\documentclass[draftcls, onecolumn]{IEEEtran}

\usepackage[normalem]{ulem} 

\usepackage{array}
\usepackage{cite}
\usepackage{tabu}
\usepackage{color}
\usepackage[normalem]{ulem}

\makeatletter
\newcommand{\thickhline}{%
    \noalign {\ifnum 0=`}\fi \hrule height 1pt
    \futurelet \reserved@a \@xhline
}
\newcolumntype{"}{@{\hskip\tabcolsep\vrule width 1pt\hskip\tabcolsep}}
\makeatother

\usepackage{hyperref}
\hypersetup{citecolor=blue,pdfborder=0 0 0,pdfstartview={FitH}}
\usepackage{amsfonts}
\usepackage{tabularx}

\ifCLASSINFOpdf
 \usepackage[pdftex]{graphicx}
\else
\fi

\usepackage{caption}
\usepackage{subcaption}

\usepackage{amsmath, amsthm, amssymb}

\newtheorem{theorem}{Theorem}
\newtheorem{lemma}[theorem]{Lemma}
\newtheorem{example}{Example}

\newtheorem{definition}{Definition}
\newtheorem{algorithm}{Algorithm}

\begin{document}

%
\title{LDPC Codes over the $q$-ary Multi-Bit Channel}

\author{Rami~Cohen,~\IEEEmembership{Graduate Student Member,~IEEE,}
Netanel~Raviv,~\IEEEmembership{Member,~IEEE,} and~Yuval~Cassuto,~\IEEEmembership{Senior Member,~IEEE}
\thanks{The first and third authors are with the Andrew and Erna Viterbi Faculty of Electrical Engineering, Technion - Israel Institute of Technology, Haifa, Israel 3200003 (email: rc@campus.technion.ac.il, ycassuto@ee.technion.ac.il); the second author is with the Department of Computer Science, Technion - Israel Institute of Technology, Haifa, Israel 3200003 (email: netanel@cs.technion.ac.il) }
\thanks{Parts of this work were presented at the 9th International Symposium on Turbo Codes \& Iterative Information Processing (ISTC), September 2016, Brest, France.}}

%



\maketitle

\begin{abstract}
In this paper, we introduce a new channel model we term the $q$-ary multi-bit channel (QMBC). This channel models a memory device, where $q$-ary symbols ($q=2^s$) are stored in the form of current/voltage levels. The symbols are read in a measurement process, which provides a symbol bit in each measurement step, starting from the most significant bit. An error event occurs when not all the symbol bits are known. To deal with such error events, we use GF($q$) low-density parity-check (LDPC) codes and analyze their decoding performance. We start with iterative-decoding threshold analysis, and derive optimal edge-label distributions for maximizing the decoding threshold. We later move to finite-length iterative-decoding analysis and propose an edge-labeling algorithm for improved decoding performance. We then provide finite-length maximum-likelihood decoding analysis for both the standard non-binary random ensemble and LDPC ensembles. Finally, we demonstrate by simulations that the proposed edge-labeling algorithm improves finite-length decoding performance by orders of magnitude.
\end{abstract}

\IEEEpeerreviewmaketitle
\section{Introduction}

In multi-level memories, information is often stored in the form of $q=2^s$ (for some integer $s$) voltage/current levels. As an example, flash memory chips with triple-level cell (TLC) technology store eight levels in each cell. In the read process, the stored levels are measured and converted to a $q$-ary symbol. In this work, we introduce the \textit{$q$-ary multi-bit channel} (QMBC) model for reading information from memory devices and modeling possible errors. The QMBC is a special case of a \textit{partial-erasure channel} \cite{CC}, where the channel output is a \textit{set} containing the input symbol.

In the QMBC, each $q$-ary symbol is decomposed into $s$ bits with hierarchical structure. The bits are organized such that when the channel erases bit $j\in\{1,\ldots,s\}$, all lower bits $\{1,\ldots,j-1\}$ are erased as well. That is, the QMBC directly models a readout by a binary-search sequence that may terminate while the last $j$ measurements are missing. In addition, the QMBC mimics errors with magnitude limits (common in non-volatile memories), where an error may affect only the $j$ lower bits of the symbol. One use of this channel is when the level-measurement process can return partial-precision read values. Another use is as a loyal and theoretically manageable proxy for designing LDPC codes for graded-magnitude errors, similarly to binary erasures being a good proxy for symmetric bit errors.

In the QMBC, the channel outputs either the input symbol, or a set of $2^j$ ($j\in\{1,...,s\}$) \textit{consecutive} symbols that contain the input symbol. In the latter case, we say that a \textit{partial-erasure} event occurred. For example, in the highest-severity partial-erasure event that is not a full erasure, the output set contains either the lower or upper $q/2$ symbols. This model is different from the $q$-ary partial-erasure channel (QPEC) model \cite{CC}, where the channel output is a random set containing the input symbol.

To deal with QMBC partial-erasure events, we use GF($q$) low-density parity-check (LDPC) codes \cite{Gallager1, MacKay} due to their low complexity of implementation and good performance under iterative decoding. We show that messages exchanged in the iterative-decoding process have certain structural properties that facilitate decoding-performance analysis. To obtain a suitable measure of asymptotic iterative-decoding performance, we extend the binary erasure channel (BEC) decoding threshold \cite{mct}, by defining the QMBC \textit{decoding threshold region}. We use the structure of the messages to both simplify the decoding-threshold region analysis and to derive an optimal code-graph edge label distribution for maximal performance.

We later move to design and analysis of {\em finite-length} LDPC codes for the QMBC. When iterative decoding is applied over the QMBC, in addition to the stopping sets \cite{Di}, the finite-length performance depends strongly on the edge labels. We theoretically characterize this dependence by analyzing the algebraic structure of the partial-erasure sets within the finite field, and propose an edge-labeling algorithm that considerably mitigates the harmful effect of stopping sets. In that, our work extends previous label-optimization algorithms (e.g., \cite{Bazarsky, Amiri}) to the special structure of the QMBC. The advantage here is that the QMBC has strong solvability conditions that are local to a single check, and thus allow neutralizing stopping sets even without relying on the cycle structure of the graph. A very interesting result we show on local solvability is the existence of  {\em universal} edge labels, which guarantee solvability at the check node for all combinations of two QMBC partial-erasure sizes that satisfy $j_1+j_2\leq s$. This generalizes the known capability of the check to resolve $s$ bits of one erased variable node to resolving any combination of $s$ bits in two partially-erased variable nodes. We then study the QMBC finite-length maximum-likelihood decoding performance, both for the standard non-binary ensemble and regular LDPC ensembles. Because QMBC erasures are {\em subsets} of the field GF$(q)$, the main analytical challenge here is in losing the linear structure. Finally, simulation results show that our edge-labeling algorithm offers significant improvement over uniform labeling, and even more so compared to using a binary LDPC code.

%

This paper is structured as follows. In Section \ref{prelim}, the QMBC model and an iterative message-passing decoder are provided. Structural properties of the iterative decoder are given in Section \ref{sec:structural}. The QMBC decoding-threshold region and optimal edge-label distributions are introduced in Section \ref{sec:de_equations}. Finite-length analysis of iterative-decoding performance and an edge-labeling algorithm for improved decoding performance are presented in Section \ref{sec:ss}. We study finite-length \textit{maximum-likelihood} decoding performance in Section \ref{sec:finite}. Finally, simulation results are presented in Section \ref{sec:sim_results} and conclusions are provided in Section \ref{sec:conclusion}.

\section{Channel Model and Iterative Decoder}
\label{prelim}

The $q$-ary multi-bit channel (QMBC) belongs to the class of partial-erasure channels \cite{CC}, where the read process provides either the correct symbol or a \textit{partially-erased} symbol. In the latter case, a subset of the input symbols that contains the correct symbol is provided as the channel output. The binary and the $q$-ary erasure channels (BEC and QEC) are special cases of the QMBC, where \textit{full} erasures may occur, carrying no non-trivial information.

\subsection{Channel model and capacity}
\label{sec:channel_model}
The QMBC input alphabet consists of $q=2^s$ symbols: $\mathcal{X}=\left\{ {0,1,...,q - 1} \right\}$, for some integer $s$. For each input symbol $x$ and $j=0,1,2,...,s$, a \textit{partial-erasure} event of type $j$ occurs when only the $s-j$ left bits of $x$ in binary representation are known. In this case, the channel output is a set of $2^j$ consecutive symbols that have the same $s-j$ left bits as $x$. We denote this output set by ${\mathcal{M}_x^j}$. Note that $x \in \mathcal{M}_x^j$ for any $j$, i.e., the correct input symbol belongs to the output set. In addition, the input symbol is completely known when $j=0$. The transition probabilities governing the QMBC are:
\begin{equation}
\label{tran_matrix}
\Pr \left( {\left. {Y = \mathcal{M}_x^j} \right|X = x} \right) = {\varepsilon _j},
\end{equation}
where $\varepsilon_j$ for $j=0,1,...,s$ are the partial-erasure probabilities. Note that for $\varepsilon_1=\varepsilon_2=...=\varepsilon_{s-1}=0$ the QMBC reduces to the QEC, and when $s=1$ the QMBC reduces to the BEC.

\begin{example}
Assume that $q=4$. Then ${\cal M}_0^1 = {\cal M}_1^1 = \left\{ {0,1} \right\},{\cal M}_2^1 = {\cal M}_3^1 = \left\{ {2,3} \right\},{\cal M}_0^2 = {\cal M}_1^2 = {\cal M}_2^2 = {\cal M}_3^2 = \left\{ {0,1,2,3} \right\}$.
\end{example}

We now move to provide the QMBC capacity.
\begin{theorem}
\label{th:capacity}
The QMBC capacity is
\begin{equation}
\label{qmbc_capacity}
1 - \sum\limits_{j = 1}^s {\frac{{j{\varepsilon _j}}}{s}},
\end{equation}
measured in $q$-ary symbols per channel use.
\end{theorem}
The proof of this theorem is provided in Appendix \ref{proof:capacity}. If the only non-zero partial-erasure probability is $\varepsilon_s$, the QMBC capacity reduces to the QEC capacity $1-\varepsilon_s$, as expected.

\subsection{GF($q$) representation}
\label{sec:gfq_rep}

For analysis purposes, we map the symbols in $\mathcal{X}$ to GF($q=2^s$) elements. Consider a basis $\{\omega_1,\omega_2,...,\omega_s\}$ of GF($q=2^s$) over GF($2$). Denote by $\left\langle {{\omega _1},{\omega _2},...,{\omega _j}} \right\rangle $ the span of the basis elements $\omega_1,\omega_2,...,\omega_j$ for $j=1,2,...,s$.  As an example, $\left\langle {{\omega _1},{\omega _2}} \right\rangle  = \left\{ {a \cdot {\omega _1} + b \cdot {\omega _2}:a,b \in \left\{ {0,1} \right\}} \right\}$. We map the sets $\mathcal{M}_0^j$ for $j=1,2,...,s$ to $\left\langle {{\omega _1},{\omega _2},...,{\omega _j}} \right\rangle $, which are \textit{subgroups} of the additive group of GF($q$). These subgroups are linear subspaces of the field GF($q=2^s$) when viewed as a dimension-$s$ vector space over GF($2$). More generally, for each $j=1,2,...,s$ and $x \in \mathcal{X}$ we map $\mathcal{M}_x^j$ to the $2^{s-j}$ cosets of $\left\langle {{\omega _1},{\omega _2},...,{\omega _j}} \right\rangle$, where the coset representatives are taken from $\left\langle {{\omega _{j+1}},{\omega_{j+2}},...,{\omega_s}} \right\rangle$.

\begin{example}\upshape
\label{ex1}
Let $\alpha$ designate a root of the primitive polynomial $x^2 +x+1$ such that $\{1,\alpha\}$ is a basis of GF($4$) over GF($2$). The sets ${\cal M}_0^0,{\cal M}_0^1$ and ${\cal M}_0^2$ are mapped to the subgroups $\{ 0\},\{ 0,1\}$ and $\{ 0,1,\alpha ,\alpha +1\}$, respectively. The cosets of $\{ 0,1\}$ are $\{ 0,1\}$ and $\{ \alpha,\alpha+1\}$. Thus, ${\cal M}_1^1$ is mapped to $\{0,1\}$, while ${\cal M}_2^1$ and ${\cal M}_3^1$ are mapped to $\{\alpha,\alpha+1\}$. 
\end{example}

We will assume a mapping as above, and will refer to symbol/field representation of the elements in $\mathcal{X}$ interchangeably.

\subsection{GF($q$) LDPC codes}
\label{LDPC_Q}

The error-correcting codes we consider for dealing with the QMBC are GF($q$) LDPC codes \cite{Gallager1, MacKay}. These codes are defined by a sparse parity-check matrix with elements taken from GF($q$). This matrix is commonly visualized as a Tanner graph \cite{tanner}. The graph is bipartite, with \textit{variable} (left) nodes corresponding to codeword symbols, and \textit{check} (right) nodes corresponding to parity-check equations. The edge labels on the graph are taken from the non-zero elements of GF($q$). The parity-check equation induced by check node $\mathtt{c}$ is $\sum\limits_{\mathtt{v} \in \mathcal{N}\left(\mathtt{c} \right)} {{h_{\mathtt{c},\mathtt{v}}}\cdot {\mathtt{v}}}  = 0$, where $\mathcal{N}(\mathtt{c})$ is the set of variable nodes adjacent to check node $\mathtt{c}$ and $h_{\mathtt{c},\mathtt{v}}$ is the label on the edge connecting check node $\mathtt{c}$ to its neighbour $\mathtt{v} \in \mathcal{N}(\mathtt{c})$. The calculations are performed using the GF($q$) arithmetic. 


LDPC codes are usually characterized by the {\it degree distributions} of the variable and check nodes. They are called {\it regular} if both variable nodes and check nodes have constant degree. Otherwise, they are called \textit{irregular}. Denote by $d_v$ and $d_c$ the maximal degree of variable nodes and check nodes, respectively. As is customary \cite{mct}, we define the degree-distribution polynomials ${\lambda \left( x \right) = \sum\limits_{i = 2}^{{d_v}} {{\lambda _i}{x^{i - 1}}} }$ and ${\rho \left( x \right) = \sum\limits_{i = 2}^{{d_c}} {{\rho _i}{x^{i - 1}}} }$, where a fraction $\lambda_i$ ($\rho_i$) of the edges is connected to variable (check) nodes of degree $i$. The {\it design rate} $R$ of an LDPC code, measured in $q$-ary symbols per channel use, is \cite{mct}:
\begin{equation}
\label{LDPC_rate}
R = 1 - \left( {\sum\limits_{i = 2}^{{d_c}} {{\rho _i}/i} } \right)/\left( {\sum\limits_{i = 2}^{{d_v}} {{\lambda _i}/i} } \right).
\end{equation}
The design rate equals the actual rate if the rows of the LDPC code parity-check matrix are linearly independent. Otherwise, the design rate is a lower bound on the actual rate.

\subsection{Set iterative decoder}

Since the QMBC belongs to the class of partial-erasure channels, we use the iterative decoder suggested for such channels in \cite{CC}. In this decoder, sets of symbols are exchanged as messages in the decoding process. The set iterative decoder extends the BEC iterative decoder \cite{mct} to partial erasures, as follows. As usual, we have \textit{variable-to-check} (VTC) and \textit{check-to-variable} (CTV) messages. We denote by ${\rm CTV}_{\mathtt{c} \to \mathtt{v}}^{\left( l \right)}$ the CTV message from check node $\mathtt{c}$ to variable node $\mathtt{v}$ at iteration $l$. In a similar way, ${\rm VTC}_{\mathtt{v} \to \mathtt{c}}^{\left( l \right)}$ denotes the VTC message at iteration $l$. Both the VTC and CTV messages are \textit{sets} containing GF($q$) elements.

An outgoing message from a graph node to a target (adjacent) node depends on incoming messages along edges connected to the source node except the outgoing message edge. At iteration $l=0$ (initialization), variable node $\mathtt{v}$ sends its channel-information set (which can be one of the sets $\mathcal{M}_x^j$ defined in Section \ref{sec:channel_model}) to adjacent check nodes. We denote these initial messages by ${\rm VTC}_{\mathtt{v} }^{\left(0 \right)}$.

A CTV message ${\rm CTV}_{\mathtt{c} \to \mathtt{v}}^{\left( l \right)}$ contains all the possible symbol values of $\mathtt{v}$ that satisfy the parity-check equation at $\mathtt{c}$ given the VTC messages to $\mathtt{c}$ at iteration $l-1$. To calculate the CTV messages efficiently, the \textit{sumset} operation \cite{Tao} is used. This operation is defined for two sets $\mathcal{A}$ and $\mathcal{B}$ that contain GF($q$) elements as
\begin{equation}
\label{sumset_def}
\mathcal{A} + \mathcal{B} \triangleq \left\{ {a + b:a \in \mathcal{A},b \in \mathcal{B}} \right\},
\end{equation}
where the addition is performed using the GF($q$) arithmetic. That is, the set $\mathcal{A} + \mathcal{B}$ contains all pairwise sums of elements taken from $\mathcal{A}$ and $\mathcal{B}$. The CTV message from check node $\mathtt{c}$ to variable node $\mathtt{v}$ is then:
\begin{equation}
\label{CTV_def}
{\rm CTV}_{\mathtt{c} \to \mathtt{v}}^{\left( l \right)} =  \sum\limits_{\mathtt{v}' \in \left\{ {\mathcal{N}\left( \mathtt{c} \right)\backslash \mathtt{v}} \right\}} {\left( {\frac{{{h_{\mathtt{c},\mathtt{v}'}}}}{{{{h_{\mathtt{c},\mathtt{v}}}}}}} \right) \cdot {{\rm VTC}_{\mathtt{v}' \to \mathtt{c}}^{\left( l-1 \right)}}},
\end{equation}
where the sum is a sumset operation and the multiplications are performed element-wise. Once all the CTV messages are calculated, the VTC messages are calculated as the \textit{intersection} of the channel-information set and the incoming CTV message sets:
\begin{equation}
\label{VTC}
{\rm VTC}_{\mathtt{v} \to \mathtt{c}}^{\left( l \right)} = {\rm VTC}_{\mathtt{v}}^{\left( 0 \right)}\bigcap {\left\{ {\bigcap\limits_{\mathtt{c}' \in \left\{ {{\cal N}\left( \mathtt{v} \right)\backslash \mathtt{c}} \right\}} {{\rm CTV}_{\mathtt{c}' \to \mathtt{v}}^{\left( l \right)}} } \right\}}.
\end{equation}
A decoding failure occurs if unresolved variable nodes (i.e., containing sets with more than one symbol) remain after the decoder terminates.

\section{Structural Properties of Exchanged Messages}
\label{sec:structural}

In this section, we show that the VTC and CTV messages admit structural properties that facilitate iterative-decoding performance analysis. Denote the additive group of GF($q$) by $\text{GF}^{+}$($q$). We will see that to analyze the probability of decoding failure, it suffices to consider messages that are subgroups of $\text{GF}^{+}$($q$). Assuming the all-zero codeword, the decoding process starts with the channel-information sets $\mathcal{M}_0^j$ as \textit{channel subgroups}, which evolve into more general subgroups in the message-passing process. We start with two fundamental properties of the sumset and intersection operations between \textit{cosets} of subgroups. Note that sums involving sets are interpreted as sumsets (see \eqref{sumset_def}).

\begin{lemma}
\label{lem:fund_prop}
Consider two subgroups $\mathcal{H}_a, \mathcal{H}_b$ of $\text{GF}^{+}$($q$) and two cosets $\mathcal{H}_{a} + g_a$ and $\mathcal{H}_b + g_b$ for some $g_a,g_b \in \text{GF}^{+}$($q$). Then
\begin{equation}
\label{sumset_property}
\left( {{{\cal H}_a} + {g_a}} \right) + \left( {{{\cal H}_b} + {g_b}} \right) = \left( {{{\cal H}_a} + {{\cal H}_b}} \right) + \left( {{g_a} + {g_b}} \right).
\end{equation}
In addition, if both cosets contain an element $\gamma$, then
\begin{equation}
\label{intersection_property}
\left( {{\mathcal{H}_a} + {g_a}} \right)\bigcap {\left( {{\mathcal{H}_{b}} + {g_{b}}} \right)}  = \left({{\mathcal{H}_a}\bigcap {{\mathcal{H}_{b}}} } \right) + \gamma.
\end{equation}
\end{lemma}
\begin{proof}
The relation in \eqref{sumset_property} is due to the associativity of the field addition operation. In addition, the sumset of $\mathcal{H}_a + \mathcal{H}_b$ forms a group, due to the closure of $\mathcal{H}_a$ and $\mathcal{H}_b$. Thus, the right-hand side of \eqref{sumset_property} is a \textit{coset} of $\mathcal{H}_a + \mathcal{H}_b$. To prove \eqref{intersection_property}, note that if $\gamma$ belongs to $\mathcal{H}_a + g_a$ (resp. $\mathcal{H}_b + g_b$) then $\mathcal{H}_a+g_a = \mathcal{H}_a + \gamma$ (resp. $\mathcal{H}_b+g_b = \mathcal{H}_b + \gamma$). An element $\mu$ lies in $\left( {{{\cal H}_a} + {g_a}} \right)\bigcap {\left( {{{\cal H}_b} + {g_b}} \right)} = \left( {{{\cal H}_a} + \gamma } \right)\bigcap {\left( {{{\cal H}_b} + \gamma } \right)}$ if and only if there are $h_a \in \mathcal{H}_a$ and $h_b \in \mathcal{H}_b$ such that $\mu = h_a + \gamma = h_b + \gamma$. This holds if and only if $\mu - \gamma = h_a = h_b$, meaning that $\mu - \gamma \in {{\cal H}_a}\bigcap {{{\cal H}_b}}$ or $\mu \in \left({{\cal H}_a}\bigcap {{{\cal H}_b}}\right) + \gamma$.
\end{proof}

As a result of Lemma \ref{lem:fund_prop}, the right-hand side of \eqref{sumset_property} is a coset of the group $\mathcal{H}_a+\mathcal{H}_b$ and the right-hand side of \eqref{intersection_property} is a coset of the group ${\mathcal{H}_a}\bigcap {{\mathcal{H}_b}}$. That is, the sumset and non-empty intersection operations between cosets result in cosets. Moreover, these operations can be performed between the underlying subgroups, followed by the addition of a constant. We leverage this observation to derive structural properties of the exchanged messages in the set iterative decoder. 

\begin{lemma}
\label{lem:cosets}
The VTC and CTV messages exchanged in the QMBC iterative-decoding process are cosets of subgroups of $\text{GF}^{+}$($q$).
\end{lemma}
\begin{proof}
As we saw in Section \ref{sec:gfq_rep}, the sets $\mathcal{M}_{0}^j$ ($j=0,1,...,s$) are mapped to \textit{subgroups} of $\text{GF}^{+}$($q$). More generally, the channel-information sets $\mathcal{M}_{x}^j$ for $x \in \mathcal{X}$ are mapped to \textit{cosets} of these subgroups. Denote by $x_\mathtt{v}$ the correct codeword symbol at a certain variable node $\mathtt{v}$. The CTV message from an adjacent check node $\mathtt{c}$ to $\mathtt{v}$ at iteration $1$ has the form (see \eqref{CTV_def})
\begin{equation}
\label{CTV_eq}
\sum\limits_{\mathtt{v}' \in \left\{ {\mathcal{N}\left( \mathtt{c} \right)\backslash \mathtt{v}} \right\}}  {\left({{g_{\mathtt{v}'}} \cdot \mathcal{M}_0^{{j_{\mathtt{v}'}}} + {g_{\mathtt{v}'}} \cdot {x_{\mathtt{v}'}}} \right)},
\end{equation}
where for each $\mathtt{v}' \in \{\mathcal{N}\left(\mathtt{c}\right) \backslash \mathtt{v} \}$, $g_{\mathtt{v}'}$ is a constant determined by the graph edge labels and $2^{j_{\mathtt{v}'}}$ is the cardinality of the channel-information set at $\mathtt{v}'$. For each $\mathtt{v}'$, the set ${g_{\mathtt{v}'}} \cdot \mathcal{M}_0^{j_{\mathtt{v}'}}$ is a subgroup of $\text{GF}^{+}$($q$), where closure follows from the closure of the subgroup $\mathcal{M}_0^{j_{{\mathtt{v}'}}}$. Therefore, \eqref{CTV_eq} is a sumset of cosets, resulting in a coset (see the first part of Lemma \ref{lem:fund_prop}).

Recall that the correct codeword symbol ${x_{\mathtt{v}}}$ is contained in any CTV message to $\mathtt{v}$, as the channel may introduce partial erasures but no errors. Thus, the sumset of cosets \eqref{CTV_eq} can be written as
\begin{equation}
\label{CTV_eq2}
\left( \sum\limits_{\mathtt{v}' \in \left\{ {\mathcal{N}\left( \mathtt{c} \right)\backslash \mathtt{v}} \right\}} {g_{\mathtt{v}'}} \cdot \mathcal{M}_0^{{j_{{\mathtt{v}'}}}} \right)+ {x_\mathtt{v}}.
\end{equation}
The VTC message at iteration $1$ from $\mathtt{v}$ to $\mathtt{c}$ is the intersection between the channel-information set at $\mathtt{v}$ and the CTV message sets from $\left\{ {{\cal N}\left( \mathtt{v} \right)\backslash \mathtt{c}} \right\}$ to $\mathtt{v}$. Both types of sets were shown above to be cosets, and all of them contain the correct codeword symbol ${x_{\mathtt{v}}}$. According to the second part of Lemma \ref{lem:fund_prop}, the intersection between these cosets is a coset. Repeating the arguments above for the next decoding iterations, an invariant is maintained that the VTC and CTV messages are cosets of subgroups of $\text{GF}^{+}$($q$).
\end{proof}

In the following theorem, we provide an important simplification for iterative-decoding performance analysis.

\begin{theorem}
\label{thm:all_zero}
The probability of decoding failure is independent of the transmitted codeword. Furthermore, if the all-zero codeword was transmitted, the exchanged messages are subgroups of $\text{GF}^{+}$($q$).
\end{theorem}

\begin{proof}
We formally prove the intuitive fact that decoding progress only depends on the underlying {\em subgroups} exchanged in the messages, and not on which cosets of these subgroups are exchanged. A VTC message from variable node $\mathtt{v}$ depends on the intersection of cosets as in \eqref{CTV_eq2}. However, an intersection of cosets is a coset of the intersection of the underlying subgroups (Lemma \ref{lem:fund_prop}). Thus, the \textit{cardinality} of the VTC message depends on the underlying subgroups $\mathcal{M}_0^{j_{\mathtt{v}}}$ only. In other words, it depends on the \textit{partial-erasure pattern}, i.e., on the cardinalities of the channel-information sets. Thus, the VTC message cardinalities are \textit{independent} of the actual transmitted codeword.

A decoding failure occurs if a variable node set cardinality is larger than one at the end of the decoding process (recall that the correct symbol is always contained in the messages). Thus, the probability of decoding failure is independent of the transmitted codeword. If the all-zero codeword is transmitted, ${x_\mathtt{v}}$ in \eqref{CTV_eq2} are all zero. Thus, the CTV messages are obtained as a sumset of \textit{subgroups}, resulting in subgroups. As a consequence, the intersection operation at variable nodes is performed between subgroups, resulting in subgroups as well.
\end{proof}

\subsection{Complexity}

The complexity of the iterative-decoding performance analysis depends on the size of the space of possible messages. Due to Theorem \ref{thm:all_zero}, the space of possible messages is upper bounded by the number of subgroups of $\text{GF}^{+}$($q$).

\begin{theorem}
\label{thm:msg_num}
The number of possible VTC and CTV messages passed in the decoding process, assuming that the all-zero codeword was transmitted, is upper bounded by
\begin{equation}
\label{eqn:num_subgroups}
T = \sum\limits_{j = 0}^s {\left( {\frac{{\prod\limits_{i = 1}^j {\left( {{2^s} - {2^{i - 1}}} \right)} }}{{\prod\limits_{i = 1}^j {\left( {{2^j} - {2^{i - 1}}} \right)} }}} \right)},
\end{equation}
which is the number of subgroups of $\text{GF}^{+}$($q$).
\end{theorem}

Note that the number of subgroups of $\text{GF}^{+}$($q$) of cardinality $2^j$ is the $j^{\rm th}$ summand in \eqref{eqn:num_subgroups}, which is the Gaussian coefficient ${s \choose j}_2$. The proof of Theorem \ref{thm:msg_num} is based on representing $\text{GF}^{+}$($q$) as an $s$-dimensional vector space over GF($2$). Then, the number of subgroups of order $2^j$ is found as the number of subspaces of dimension $2^j$ (see e.g. \cite{Prasad} for the details). We remark that the actual number of subgroups exchanged in the decoding process (assuming that the all-zero codeword was transmitted) is not necessarily $T$. Instead, it depends on the channel information and on the edge labels. As an example, the only possible subgroups in the full-erasure case (i.e., if the only non-zero partial-erasure probability is $\varepsilon_s$) are $\mathcal{M}_0^0 =\left\{ 0 \right\}$ and $\mathcal{M}_0^s$, where the latter set contains all the field elements.

The number of subgroups of $\text{GF}^{+}$($q$) is plotted in Figure \ref{fig:num_subgroups} compared to the number of non-empty subsets of $\text{GF}^{+}$($q$) as a reference. This figure reveals the importance of the QMBC iterative-decoder structure to the analysis feasibility, by which the number of subgroups is orders of magnitude smaller compared to the number of subsets of $\text{GF}^{+}$($q$). Hence performing density-evolution analysis for the QMBC is orders of magnitude less complex than for a general channel in the class of partial-erasure channels.

\begin{figure}[t!]
\centering
\includegraphics[scale=0.6]{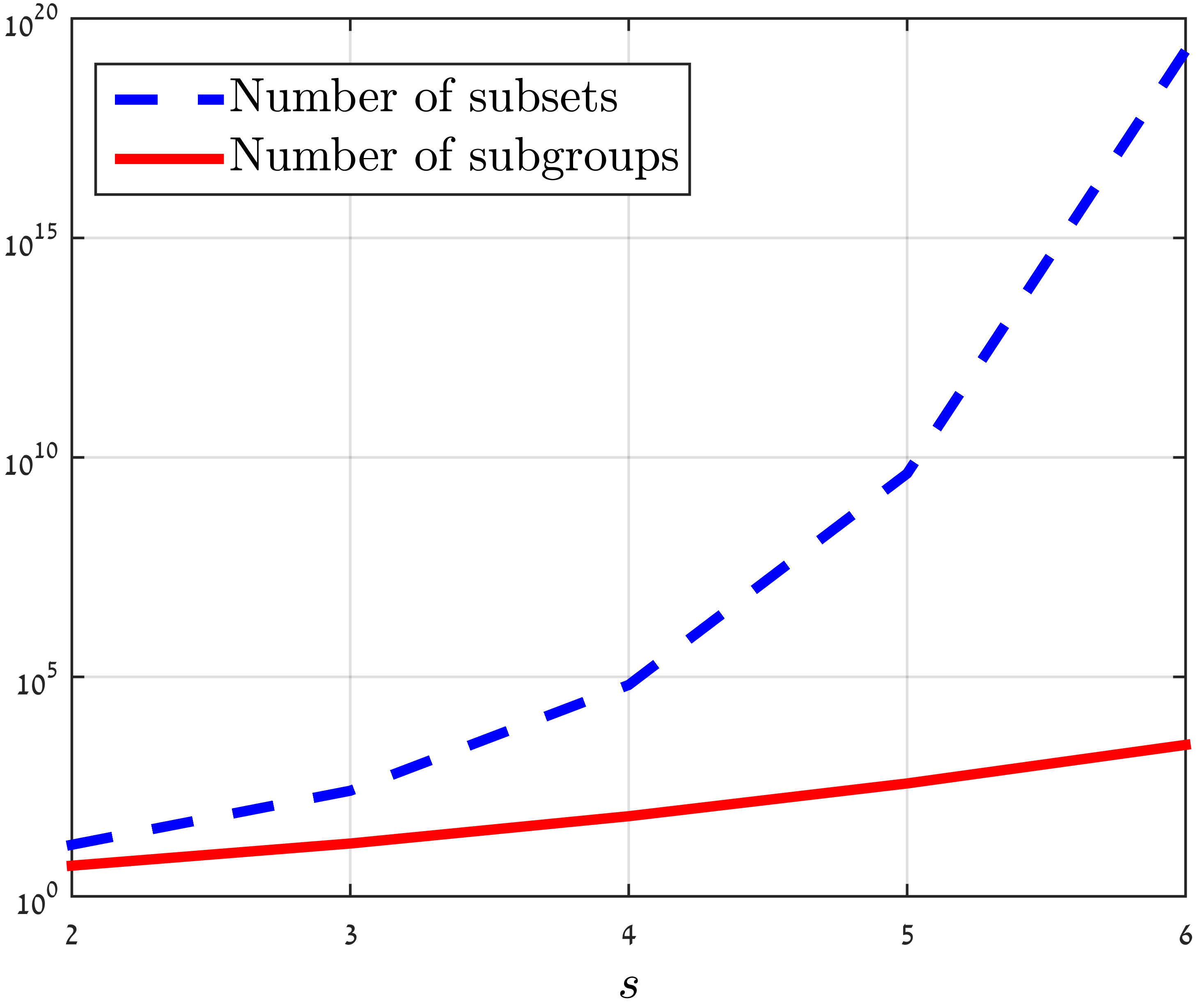}
\caption{The number of subgroups of $\text{GF}^{+}$($q$) compared to the number of subsets.}
\label{fig:num_subgroups}
\end{figure}


\section{The QMBC Decoding Threshold Region}
\label{sec:de_equations}

To evaluate the performance of the iterative decoder, we use the density evolution method \cite{RU, Rathi, bennatan}. In this method, the probabilities of the exchanged messages as a function of the decoding iteration are tracked. The code length is assumed to be sufficiently large, such that the exchanged messages are statistically independent with high probability \cite{RU}. Let us consider a Tanner graph drawn uniformly at random out of the graphs with certain degree distributions $\lambda(x)$ and $\rho(x)$. The transmission of the all-zero codeword is assumed (see Theorem \ref{thm:all_zero}), such that the possible messages are subgroups of $\text{GF}^{+}$($q$). We denote these subgroups by $\left\{ {{\mathcal{H}_t}} \right\}_{t = 1}^T$ (recall that $T$ is provided in \eqref{eqn:num_subgroups}). For convenience, we assume that $\mathcal{H}_1 = \mathcal{M}_0^0 = \{0\}$.
\begin{example} \upshape
\label{ex:subgroups_order}
Consider the representation of GF($4$) in Example \ref{ex1}. There are $T=5$ subgroups of $\text{GF}^{+}$($4$), which can be ordered as follows: ${\mathcal{H}_1} = \left\{ 0 \right\}$, ${\mathcal{H}_2} = \left\{ {0,1} \right\},{\mathcal{H}_3} = \left\{ {0,\alpha} \right\},{\mathcal{H}_4} = \left\{ {0,\alpha+1} \right\}$ and $\mathcal{H}_5 = \left\{ {0,1,\alpha,\alpha+1} \right\}$.
\end{example}

To obtain the QMBC density-evolution equations, we define $w_t^{\left( l \right)}$ (resp. $z_t^{\left( l \right)}$) as the probability that a CTV (resp. VTC) message at iteration $l$ is $\mathcal{H}_t$. We denote by $\mathcal{M}_{i-1}$ an ordered list containing $i-1$ subgroup indices taken from $\left\{ {1,2,...,T} \right\}$. These subgroups are interpreted as VTC (resp. CTV) messages to a check (resp. variable) node of degree $i$.

\begin{example} \upshape
Assume that $q=4$ (i.e., $T=5$ subgroups) and consider the $(3,6)$ LDPC code ensemble. Then $\mathcal{M}_2$ can be one of the ordered lists $\left[ {1,1} \right], \left[ {1,2} \right],\ldots,\left[ {5,5} \right]$. Similarly, $\mathcal{M}_5$ can be one of the ordered lists  $\left[ {1,1,1,1,1} \right]$, $\left[ {1,1,1,1,2} \right]$,$\ldots$,$\left[ {5,5,5,5,5} \right]$.
\end{example}

In the case of binary LDPC codes, the edge labels of a Tanner graph are simply '$1$'s. In the GF($q$) case, they are taken from the non-zero field elements. Thus, a GF($q$) LDPC ensemble is characterized by an edge-label distribution in addition to the degree distributions. Let us denote the edge-label probability distribution by $\mathbb{L}$. We define ${P_t}\left( {{{\cal M}_{i-1}},\mathbb{L}} \right)$ as the probability of $\mathcal{H}_t$ as a CTV message, given the VTC messages indexed in $\mathcal{M}_{i-1}$, and the distribution $\mathbb{L}$. We also define $I_{t,j}\left( {\mathcal{M}_{i-1}} \right)$ as an indicator function, which equals $1$ if the intersection of the CTV messages indexed in $\mathcal{M}_{i-1}$ and the channel-information set $\mathcal{M}_0^j$ is the VTC message $\mathcal{H}_t$. Otherwise, $I_{t,j}\left( {{\mathcal{M}_{i-1}}} \right)$ is $0$ (note that the calculation of $I_{t,j}$ is independent of the edge labels). The following density-evolution equations are obtained:
\begin{equation}
\label{de1}
w_t^{\left( l \right)} = \sum\limits_{i = 2}^{{d_c}} {{\rho _i}}
\sum\limits_{{\mathcal{M}_{i-1}}} {\left( {\prod\limits_{m \in \mathcal{M}_{i-1}} {z_m^{\left( {l - 1} \right)}} } \right)}  \cdot {P_t}\left( {\mathcal{M}_{i-1}},\mathbb{L} \right),
\end{equation}
\begin{align}
\label{de2}
z_t^{\left( l \right)} &= \sum\limits_{i = 2}^{{d_v}} {{\lambda _i}}  \sum\limits_{j = 0}^s {{\varepsilon _j}} \sum\limits_{{\mathcal{M}_{i-1}}} {\left( {\prod\limits_{m \in \mathcal{M}_{i-1}} {w_m^{\left( l \right)}} } \right)}   \cdot I_{t,j}\left( { \mathcal{M}_{i-1}} \right),
\end{align}
where the summation over $\mathcal{M}_{i-1}$ is understood over all the ordered lists containing $i-1$ subgroup indices taken from $\left\{ {1,2,...,T} \right\}$. The initial conditions of the density-evolution equations \eqref{de1}-\eqref{de2} are determined by the transition probabilities in \eqref{tran_matrix}. That is, for each $t$ such that $\mathcal{H}_t = \mathcal{M}_0^j$ ($j=0,1,...,s$), $z_t^{(0)}$ is initialized to $\varepsilon_j$. For example, if $q=4$ and the subgroups are numbered as in Example \ref{ex:subgroups_order}, then $z_1^{(0)}=\varepsilon_0$, $z_2^{(0)}=\varepsilon_1$, $z_5^{(0)}=\varepsilon_2$ and $z_3^{(0)}=z_4^{(0)}=0$. The asymptotic probability of decoding failure at iteration $l$, denoted $P_{\rm error}^{\left( l \right)}$, is the probability that a VTC message at iteration $l$ is not $\mathcal{H}_1 = \left\{ 0 \right\}$:
\begin{equation}
P_{\rm {error}}^{\left( l \right)} = \sum\limits_{i = 2}^T {z_i^{\left( l \right)}}=1-z_1^{(l)}.
\label{Pe}
\end{equation}

The QMBC is characterized by multiple partial-erasure probabilities $\left\{ {{\varepsilon _j}} \right\}_{j = 1}^s$ rather than by a single erasure probability (as in the BEC or the QEC). Thus, we define the \textit{QMBC decoding threshold region} by extending the BEC decoding threshold \cite{mct}. First, define the following \textit{QMBC $\mathbb{L}$-region} for given $\left(\lambda(x),\rho(x)\right)$ degree-distribution pair and edge-label distribution $\mathbb{L}$
\begin{equation}
\label{deL_threshold}
\Omega_{\mathbb{L}}\left( {{\lambda},{\rho}} \right) = \left\{ {{\varepsilon _1},{\varepsilon _2},...,{\varepsilon _s} \in {{\left[ {0,1} \right]}^s}:\mathop {\lim }\limits_{l \to \infty } P_{\rm error}^{\left( l \right)}(\mathbb{L}) = 0} \right\}.
\end{equation}
That is, an $\mathbb{L}$-region contains the partial-erasure probabilities leading asymptotically to zero probability of decoding failure under the edge-label distribution $\mathbb{L}$. The QMBC decoding-threshold region is the union of the QMBC $\mathbb{L}$-regions over all possible choices of $\mathbb{L}$:
\begin{equation}
\label{eq:threshold_def}
\Omega \left(\lambda,\rho \right)  = \bigcup\limits_\mathbb{L} \Omega_{\mathbb{L}}\left( {{\lambda},{\rho}} \right).
\end{equation}
If both the boundaries of $\Omega \left(\lambda,\rho \right)$ and $\Omega_{\mathbb{L}}\left( {{\lambda},{\rho}} \right)$ contain the same certain point, we say that $\mathbb{L}$ is \textit{optimal} with respect to this point.

\subsection{Optimal edge-label distributions}
\label{subsection:eld}

As mentioned earlier, GF($q$) LDPC code ensembles are characterized by edge-label probability distributions in addition to degree distributions. In the following theorem, it is demonstrated that a poor selection of label distribution may degrade performance to that of a much worse channel. Denote by ${\varepsilon_{{\rm{BEC}}}}$ the decoding threshold of the BEC (or QEC) for a given degree-distribution polynomial pair $\lambda\left(x\right)$ and $\rho\left(x\right)$.

\begin{theorem}
\label{thm:QMBC_BEC}
If the edge-label distribution $\mathbb{L}$ is chosen such that one of the non-zero GF($q$) elements appears with probability $1$ (i.e., all the labels are the same), then
\begin{equation}
\Omega _{\mathbb{L}}\left(\lambda,\rho \right)  =\left\{ {{\varepsilon _1},{\varepsilon _2},...,{\varepsilon _s} \in {{\left[ {0,1} \right]}^s}:\sum\limits_{j = 1}^s {{\varepsilon _j}}  \le {\varepsilon_{{\rm{BEC}}}}} \right\}\label{eq:MHP_threshold}.
\end{equation}
\end{theorem}
That is, when the labels are all the same, a partial erasure is asymptotically equivalent to a full erasure, which is an undesired property. The key observation in proving this theorem is that messages exchanged in this case are restricted to the channel information messages (i.e., to the initial subgroups $\mathcal{M}_0^j$). Thus, the only way to get cardinality-$1$ intersection at a variable node is when a neighbouring check node has all its other neighbours with cardinality $1$, same as when decoding over the BEC. The details are provided in Appendix \ref{proof:QMBC_BEC}. As an immediate consequence of Theorem \ref{thm:QMBC_BEC}, simply taking binary ensembles (where the edge labels are all '1') with good performance (e.g., BEC capacity-achieving) necessarily gives poor performance over the QMBC.

In the following, we derive explicitly \textit{optimal} $\mathbb{L}$ distributions for key points of interest in the QMBC decoding threshold region. For the derivation, we assume that the only non-zero partial-erasure probability is $\varepsilon_{j_{\rm max}}$. This choice does not mean that we are only interested in correcting partial erasures of type $j_{\rm max}$, but rather that we want to analyze the case when these are the dominant type of erasures. We assume a polynomial basis of GF($q$), where $\mathcal{M}_0^j$ (for $j=0,1,...,s$) contains all the polynomials of degree at most $j-1$ with coefficients in GF($2$). These polynomials are evaluated at a primitive element of GF($q$), denoted $\alpha$. In this case, a basis to GF($q$) over GF($2$) is $\{1,\alpha,\alpha^2,...,\alpha^{s-1}\}$.



\begin{theorem}
\label{thm:L_optimal}
Suppose that $j_{\max}$ divides $s$. Then choosing $\mathbb{L}$ as the uniform distribution on $\left\{ {{\alpha ^{t \cdot {j_{\max }}}}} \right\}_{t = 0}^{s/{j_{\max }} - 1}$ is optimal with respect to achieving capacity.
\end{theorem}

\begin{proof}

Suppose that the edge labels are taken from $\left\{ {{\alpha ^{t \cdot {j_{\max }}}}} \right\}_{t = 0}^{s/{j_{\max }} - 1}$. Denote the probability that a variable node is partially erased to $\mathcal{M}_0^{j_{\rm max}}$ at iteration $l$ by $y_l$, where $y_0 = \varepsilon_{j_{\rm max}}$. We claim that a CTV message to a variable node has a non-trivial intersection (i.e., containing a non-zero element) with $\mathcal{M}_0^{j_{\rm max}}$  if and only if at least one of the incoming VTC messages is a partial erasure \textit{and} the label on this incoming VTC message edge is the same as the label on the outgoing CTV message edge.

To see that, note that if the labels are the same, then $\mathcal{M}_0^{j_{\rm max}}$ is an argument in the CTV sumset operation (see~\eqref{CTV_def}), whose result must contain  $\mathcal{M}_0^{j_{\rm max}}$. Conversely, if edges from all partially-erased variable nodes have labels different from the label $h$ to the target variable node, we show that the CTV message intersects with $\mathcal{M}_0^{{j_{\max }}}$ only on $\{0\}$. Take an edge label $h_i$ of one partially-erased variable node. The labels $h,h_i\in\left\{ {{\alpha ^{t \cdot {j_{\max }}}}} \right\}_{t = 0}^{s/{j_{\max }} - 1}$ as monomials in $\alpha$ have degrees separated by at least $j_{\rm max}$. That means ${{h} \cdot \mathcal{M}_0^{{j_{\max }}}}$ and ${{h_i} \cdot \mathcal{M}_0^{{j_{\max }}}}$ intersect only on the $0$ polynomial. This is true for all $i$, and thus any sum $\sum\limits_{i} {{h_i} \cdot {x_i}}$, where $x_i$'s are elements from $\mathcal{M}_0^{{j_{\max }}}$ not all zero, gives a polynomial not in $h\cdot \mathcal{M}_0^{{j_{\max }}}$. Equivalently, the CTV message intersects with $\mathcal{M}_0^{{j_{\max }}}$ only on the symbol $0$.

Now by choosing $\mathbb{L}$ as the uniform distribution on $\left\{ {{\alpha ^{t \cdot {j_{\max }}}}} \right\}_{t = 0}^{s/{j_{\max }} - 1}$, each label has probability $j_{\rm max}/s$, and by the argument above a CTV message contains $\mathcal{M}_0^{j_{\rm max}}$ with probability
\begin{equation}
\sum\limits_{i = 2}^{{d_c}} {{\rho _i}} \left( {1 - {{\left( {1 - {y_l}\frac{{{j_{\max }}}}{s}} \right)}^{i - 1}}} \right) = 1 - \rho \left( {1 - \frac{{{y_l}}}{{s/{j_{\max }}}}} \right).
\end{equation}
The product ${y_l}\frac{j_{\max}}{s}$ is the probability that both ``bad'' events happen: the variable node connected by the incoming edge is partially erased (with probability $y_l$), and its edge has the same label as the one on the outgoing edge (with probability $\frac{j_{\max}}{s}$). The two events are statistically independent hence the product. A variable node remains partially-erased at iteration $l+1$ if and only if it was partially-erased initially (with probability $\varepsilon_{j_{\max}}$), and all its incoming CTV messages contain $\mathcal{M}_0^{j_{\rm max}}$. This leads to the \textit{single-letter} recurrence relation
\begin{align}
\label{de_sl}
&{y_{l + 1}} = {\varepsilon _{{j_{\max }}}\cdot}\lambda \left( {1 - \rho \left( {1 - \frac{{{y_l}}}{{s/{j_{\max }}}}} \right)} \right).
\end{align}
The expression in \eqref{de_sl} is the same recurrence equation as the BEC/QEC density evolution, only with $y_l$ divided by $s/j_{\max}$ in the argument of $\rho(x)$. That is, we obtained a QMBC decoding threshold that is $s/j_{\max}$ times the BEC/QEC threshold for the same ensemble (when $\varepsilon_{j_{\max}}$ is the only non-zero partial-erasure probability). This is optimal because a BEC/QEC capacity-achieving ensemble will give a capacity-achieving QMBC ensemble according to \eqref{qmbc_capacity}.
\end{proof}

We remark that as all finite fields of the same order are isomorphic, the basis elements in $\left\{ {{\alpha ^{t \cdot {j_{\max }}}}} \right\}_{t = 0}^{s/{j_{\max }} - 1}$ can always be mapped to basis elements in any other representation of GF($q$). As a consequence of Theorem \ref{thm:L_optimal}, we can calculate \textit{explicitly} the threshold of the optimal label distribution for any code ensemble, for $j_{\rm max}$ and $q$ values given in the theorem. We now demonstrate how the optimal edge-label distribution derived in Theorem \ref{thm:L_optimal} improves the decoding performance. Assume that $q=4$ and partial erasures of type $j_{\rm max}=1$. In Figure \ref{fig:capacity}, the QMBC $\mathbb{L}$-region defined in \eqref{deL_threshold} is plotted for the optimal distribution (solid line) and is compared to the uniform distribution on the non-zero field elements (dotted line), for the $(3,6)$ LDPC code ensemble. The QMBC Shannon capacity region is plotted (dashed line) for reference.

For the optimal distribution, the lower-right corner is $\varepsilon_1=0.858$, double the QEC threshold $0.429$, according to \eqref{de_sl}. At the upper-left corner ($\varepsilon_1=0$), both label distributions attain the same $\varepsilon_2$ threshold -- identical to the standard QEC threshold for full erasures. While the optimal distribution is superior at the lower-right corner, Figure~\ref{fig:capacity} reveals that there are values of $\varepsilon_2>0$ at which the uniform distribution has a higher $\varepsilon_1$ threshold. This hints that in general there is no single distribution $\mathbb{L}$ universally optimal for all combinations of $\{\varepsilon_j\}_{j=1}^{s}$.

\begin{figure}[!t]
\centering
\includegraphics[scale=0.6]{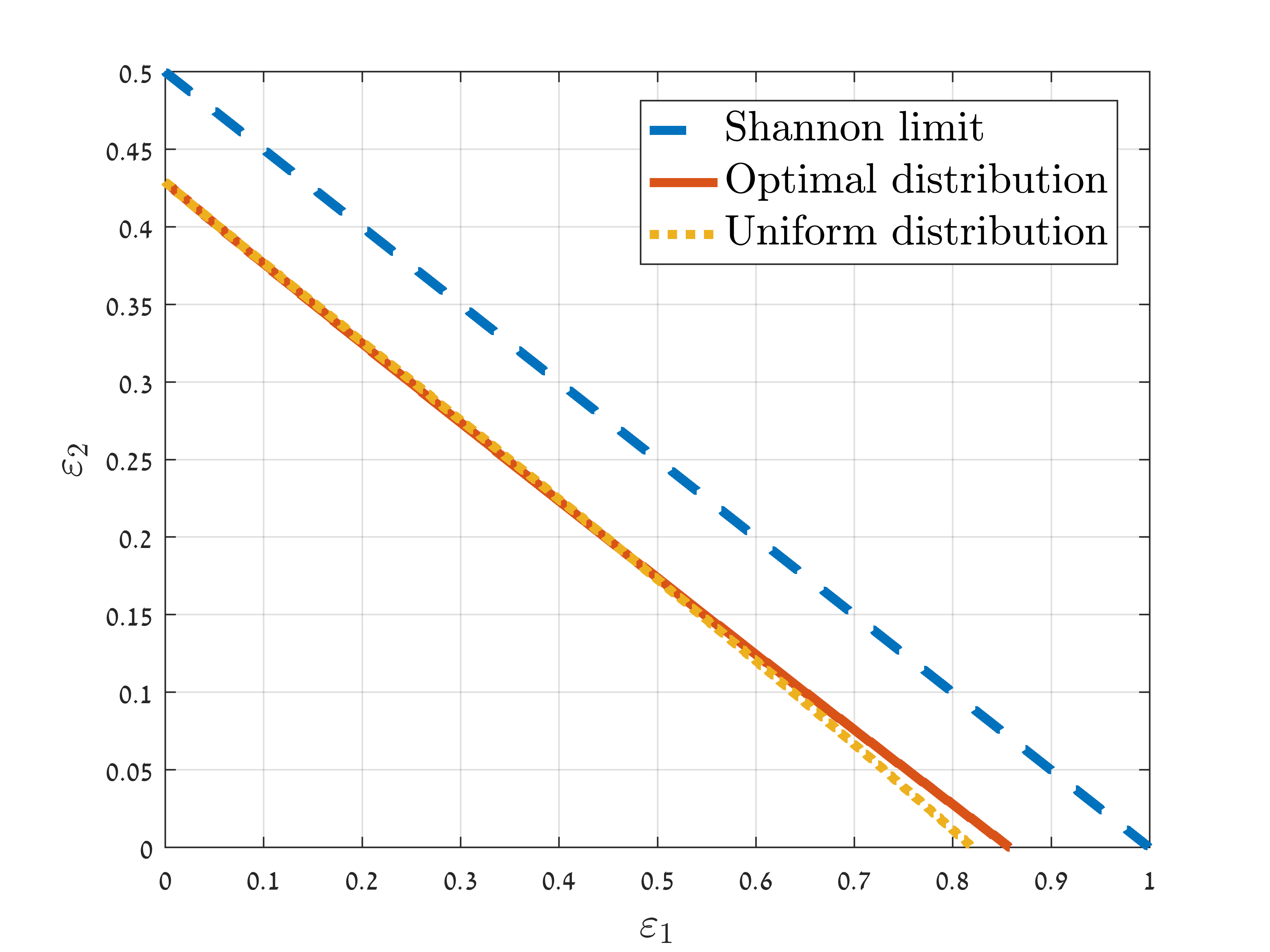}
\caption{The GF($4$) QMBC $\mathbb{L}$-regions of two edge-label distributions for the $(3,6)$ LDPC code ensemble. The QMBC Shannon capacity region is plotted for reference.}
\label{fig:capacity}
\end{figure}

It is an interesting fact that achieving optimality requires a label distribution that is {\em not} the uniform distribution on the non-zero field elements. We note that we can alternatively achieve optimality by using a binary capacity-achieving ensemble on $j_{\max}$ least significant bits of the symbols. However, the advantage of $q$-ary ensembles with an optimal edge-label distribution is that in addition to the optimality for $\varepsilon_{j_{\max}}$, the same code has good correction performance for infinitely many combinations of partial-erasure probabilities. 

\section{Edge-labeling Algorithm for Improved Finite-Length Performance}
\label{sec:ss}

In this section, we show how improved finite-length decoding performance is achieved by a wise labeling of the LDPC graph edges.

\subsection{Stopping sets and local resolvability}

A stopping set $\mathcal{S}$ is defined as a subset of variable nodes, such that all neighbours (check nodes) of $\mathcal{S}$ are connected to $\mathcal{S}$ at least twice. A key result in BEC finite-length iterative-decoding performance analysis is that the variable nodes in the maximal (fully) erased stopping set remain erased when the decoder stops \cite{Di, mct}. However, QMBC partially-erased variable nodes that belong to a stopping set might be eventually \textit{resolved}. The reason is that with partial erasures the iterative decoder can make progress even if two or more neighbours of a check node are partially erased. This is demonstrated in the following example.

\begin{example} \upshape \label{example:stopping_set}
Consider the Tanner graph in Figure \ref{fig:ss}, where the variable nodes $\mathtt{v}_1$ and $\mathtt{v}_2$ form a partially-erased stopping set ($q=4$ is assumed). The initial CTV messages from the check node at the bottom are $\{0,h_2/h_1\}$ to $\mathtt{v}_1$ and $\{0,h_1/h_2\}$ to $\mathtt{v}_2$. If $h_1=h_2$, the variable nodes are not resolved, as the intersection operation at variable nodes results in $\{0,1\}$. Otherwise, they are resolved as $\{0\}$.
\end{example}

\begin{figure}[!t]
\centering
\includegraphics[scale=0.65]{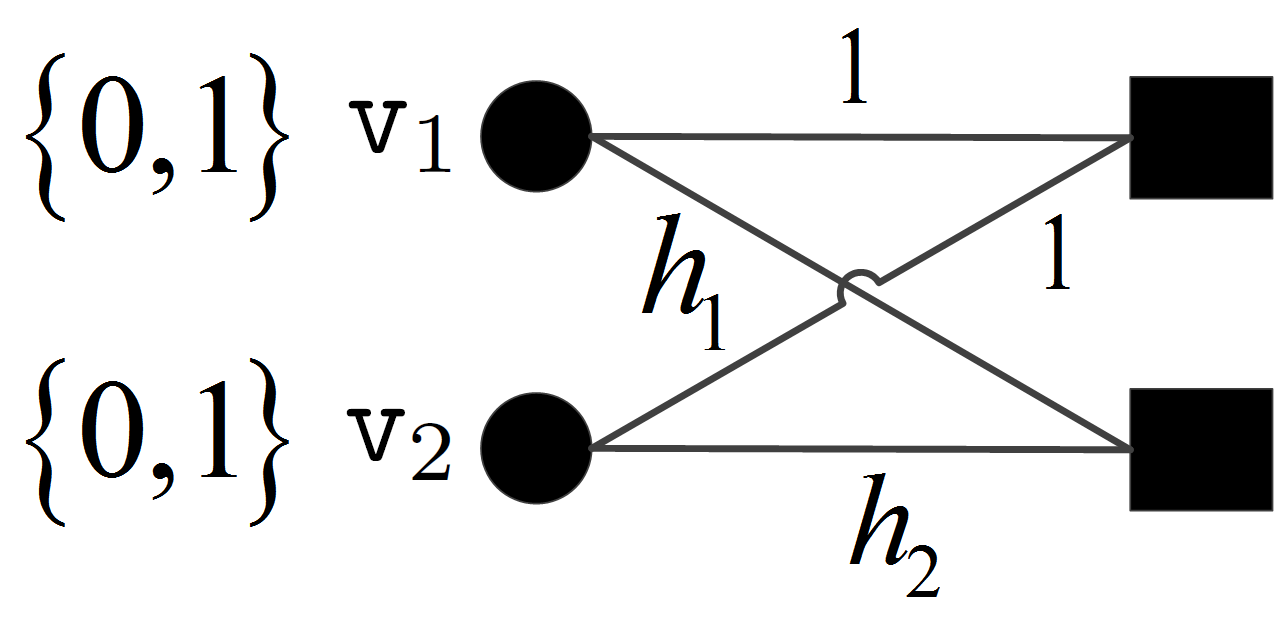}
\caption{$\mathtt{v}_1$ and $\mathtt{v}_2$ form a partially-erased stopping set (the channel information sets appear to the left). The resolvability of $\mathtt{v}_1$ and $\mathtt{v}_2$ depends on the values of $h_1$ and $h_2$.}
\label{fig:ss}
\end{figure}



As shown in Example \ref{example:stopping_set}, partially-erased variable nodes in a stopping set might be eventually resolved, depending on the edge-label configuration. However, non-resolved partial erasures must belong to a stopping set. Let us denote by $\mathcal{E}$ the set of partially-erased variable nodes.

\begin{lemma}
\label{lem:QMBC_stopping_sets}
The variable nodes that remain unresolved when the iterative QMBC decoder terminates belong to the maximum stopping set contained in $\mathcal{E}$.
\end{lemma}

\begin{proof}
Consider a variable node $\mathtt{v}$. If $\mathtt{v} \notin \mathcal{E}$ then $\mathtt{v}$ is trivially resolved. Suppose $\mathtt{v} \in \mathcal{E}$ but $\mathtt{v}$ is not in a stopping set contained in $\mathcal{E}$. In this case, it has at least one neighbouring check node whose connected variable nodes except $\mathtt{v}$ are not partially erased and $\mathtt{v}$ is resolved. Finally, if $\mathtt{v}$ belongs to a stopping set contained in $\mathcal{E}$, any neighbouring check node of $\mathtt{v}$ is connected to at least one additional partially-erased variable node. Only in this case $\mathtt{v}$ may not be resolved.
\end{proof}

Consider a check node connected to $\kappa$ partially-erased variable nodes denoted $\mathtt{v}_1, \mathtt{v}_2,...,\mathtt{v}_{\kappa}$, via edge labels $h_1,h_2,...,h_{\kappa}$, respectively. We show that there are values of the edge labels such that a decoding progress is guaranteed, independently of information from any other variable node. Recall that in the full-erasure case (i.e., BEC or QEC), the local parity-check equation at a check node resolves at most one (full) erasure. However, it is possible to resolve \textit{multiple} partial erasures in the QMBC case.

\begin{definition}
\label{def:good_labels}
The edge labels $h_1,h_2,...,h_{\kappa}$ are said to be $\kappa$-resolvable if $\mathtt{v}_1, \mathtt{v}_2,...,\mathtt{v}_{\kappa}$ are resolvable (i.e., decoded successfully), independently of other variable nodes.
\end{definition}

The motivation for Definition \ref{def:good_labels} is that by placing resolvable edge labels in stopping sets, improved decoding performance is expected. Let us denote by $j_{{\rm{max}}}$ the dominant partial-erasure type, which occurs with the partial-erasure probability $\varepsilon_{j_{\rm max}}$. Consider the basis $\{1,\alpha,\alpha^2,...,\alpha^{s-1}\}$ of GF($q$) over GF($2$) (see Section \ref{sec:de_equations}).


\begin{theorem}
\label{theorem:kappa}
Suppose that $j_{\max}$ divides $s$. The edge labels $\left\{ \alpha ^{t \cdot {j_{\max }}} \right\}_{t = 0}^{s/{j_{\max }} - 1}$ are $\left(s/j_{\rm max}\right)$-resolvable for any set of $s / j_{\max}$ variable nodes of partial-erasure type at most $j_{\max}$. In addition, there is no larger set of resolvable edge labels in this case.
\end{theorem}

\begin{proof}
Following the proof of Theorem \ref{thm:L_optimal}, if $h_i$ are distinct labels taken from $\left\{ {{\alpha ^{t \cdot {j_{\max }}}}} \right\}_{t = 0}^{s/{j_{\max }} - 1}$, the non-zero polynomials $h_i\cdot x_i$ of the variable nodes $\mathtt{v}_i$ have disjoint degrees, and thus can only satisfy the check equation if they are all zero. Hence the variable nodes can be resolved locally at the check node.  To see that no larger set of resolvable edge labels exists, note that any set of $ {s/{j_{\max }}} +1$ must contain at least two polynomials with degrees separated by less than $j_{\rm max}$. In this case, we can find a non-zero assignment to $x_1,x_2$ such that $h_1x_1+h_2x_2=0$, and the variable nodes cannot be resolved by this check.
\end{proof}


\subsection{Universal edge labeling}

In addition to partial erasures of type $j_{\rm max}$, a wider spectrum of partial erasures can be resolved when considering check nodes of degree $2$. The resolvability of variable nodes connected to such check nodes is important, as every stopping set (in graphs without singly-connected variable nodes) is comprised of cycles that contain degree-$2$ check nodes \cite{Tian}. As an example, the stopping set in Figure \ref{fig:ss} is comprised of one cycle of length $4$, with two check nodes of degree $2$. The next theorem shows that for degree-$2$ check nodes we can always find edge labels that resolve QMBC partial erasures {\em universally}, that is, the same pair of labels will resolve any partial-erasure combination $(j_1,j_2)$ satisfying $j_1+j_2\leq s$.

\begin{theorem}
\label{th:2res}
For any $q=2^s$ there exists a pair of GF$(q)$ field elements $h_1,h_2$ such that if $x\in\mathcal{M}_{x_1}^{j_1}$, $x'\in\mathcal{M}_{x_2}^{j_2}$, then $h_1 \cdot x+h_2 \cdot x'=0$ implies $x=x_1$, $x'=x_2$, for any $j_1 +j_2 \le s$.
\end{theorem}
\begin{proof}
From the subgroup structure and similarly to the zero-codeword assumption in Theorem~\ref{thm:all_zero}, we can assume without loss of generality that $x_1=x_2=0$. We now fix $h_1=1$ and prove the existence of a non-zero field element $h_2=h$ such that the only solution to the equation $x+h \cdot x' = 0$ for $x \in \mathcal{M}_0^{j_1}$ and $x' \in \mathcal{M}_0^{j_2}$ ($j_1 + j_2 \le s$) is the trivial solution $x=x'=0$. Consider a fixed but arbitrary basis $\{\omega_1,\omega_2,...,\omega_s\}$ of GF($q=2^s$) over GF($2$) (see Section \ref{sec:gfq_rep}). Such an $h$ exists if the subgroups $h \cdot \mathcal{M}_0^{j_2}$ are spanned by basis elements disjoint from the basis elements spanning $\mathcal{M}_0^{j_1}$, which guarantees that there is no dependent combination of elements from $\mathcal{M}_0^{j_1}$ and $h \cdot \mathcal{M}_0^{j_2}$. From the requirement to cover all possible $j_1,j_2:j_1 + j_2 \le s$, this is equivalent to requiring that each of the sets $\{\omega_1,\omega_2,...,\omega_{s-1},h \cdot w_1\}$, $\{\omega_1,\omega_2,...,\omega_{s-2},h \cdot \omega_1,h \cdot \omega_2\}$ $,\ldots,$ $\{w_1,h \cdot \omega_1,h \cdot \omega_2, ..., h \cdot \omega_{s-1}\}$ is a basis of GF($q$) over GF($2$). Since $\{\omega_1,\omega_2,...,\omega_s\}$ is a basis, the requirement above holds if $h \cdot {\omega _1} \notin \left\langle {{\omega _1},{\omega _2},...,{\omega _{s - 1}}} \right\rangle $, $h \cdot {\omega _1} \notin \left\langle {h \cdot {\omega _2},{\omega _1},...,{\omega _{s - 2}}} \right\rangle $ $\ldots$ $h \cdot {\omega _1} \notin \left\langle {h \cdot {\omega _2},h \cdot {\omega _3},...,h \cdot {\omega _{s - 1}},{\omega _1}} \right\rangle$.

For $h \cdot {\omega _1} \notin \left\langle {{\omega _1},{\omega _2},...,{\omega _{s - 1}}} \right\rangle$ to hold, we must discard from the candidates for $h$ (where we start with all the field elements as candidates) the ${2^{s - 1}}$ field elements in $\left\langle {1,{\omega _2}/{\omega _1},...,{\omega _{s - 1}}/{\omega _1}} \right\rangle$. For $h \cdot {\omega _1} \notin \left\langle {h \cdot {\omega _2},{\omega _1},...,{\omega _{s - 2}}} \right\rangle$ to hold, we have to discard the field elements in $\left\langle {h \cdot {\omega _2}/{\omega _1},1,...,{\omega _{s - 2}}/{\omega _1}} \right\rangle$. But, elements in $\left\langle {h \cdot {\omega _2}/{\omega _1},1,...,{\omega _{s - 2}}/{\omega _1}} \right\rangle$ obtained with the coefficient of ${h \cdot {\omega _2}/{\omega _1}}$ set to zero were already discarded. Thus, we now discard only ${2^{s - 2}}$ elements not discarded in the previous step. Continuing in a similar fashion, we remain with ${2^s} - \sum\limits_{i = 1}^{s - 1} {{2^i}}  = 2 > 0$ $h$'s satisfying the requirements, which proves existence.
\end{proof}
Note that in addition to existence, the proof of Theorem \ref{th:2res} provides a \textit{constructive} way for finding universally resolvable edge labels for degree-2 check nodes.

\subsection{Edge-labeling algorithm}
\label{ss:edge_labeling}

Based on the existence of resolvable and universally-resolvable edge labels, we propose an edge-labeling algorithm for improved finite-length decoding performance. The idea is to distribute resolvable edge labels  within edges of stopping sets such that partially-erased variable nodes are more likely to be resolved. Consider an LDPC graph with edge labels uniformly selected from the non-zero elements of GF$(q)$. Suppose that the dominant partial-erasure type is $j_{\max}$, and that $j_{\max}$ divides $s$. If $j_{\max}$ does not divide $s$, then the maximal partial-erasure type (smaller than $j_{\max}$) that divides $s$ is considered instead.

\begin{algorithm}
{(Edge labeling)}
\upshape
\label{QMBC_mitigation}
~\begin{enumerate}
\item Run the BEC iterative decoder with the channel parameter $\varepsilon = {\varepsilon_{j_{\rm max}}}$ for a predefined number of times. After each run, store the set of unresolved variable nodes. \label{step1}
\item Initialize $\Sigma$ as the subgraph induced by the variable nodes from the sets of Step \ref{step1}. Rank the variable nodes by their number of occurrences in the sets.
\item Modify the edge labels of check nodes of degree $2$ connected to variable nodes in $\Sigma$ to the universally resolvable edge labels found using Theorem \ref{th:2res}. Set the rank of connected variable nodes to $0$. \label{2res_step}
\item Modify the edge labels of check nodes in $\Sigma$ of degree larger than $2$ but no larger than $s/{j_{\max}}$ to labels taken from $\left\{ {{\alpha ^{t \cdot {j_{\max }}}}} \right\}_{t = 0}^{s/{j_{\max }} - 1}$. Set the rank of connected variable nodes to $0$. \label{kres_step}
\item Run over the sets found in Step \ref{step1} by ascending cardinality. For each check node connected to a set:
\begin{enumerate}
\item Set $\kappa'$ as the minimum between the number of non-zero ranking variable nodes and $s/{j_{\max}}$.
\item Modify the $\kappa'$ edge labels connected to non-zero highest-ranking variable nodes according to either Step \ref{2res_step} (if $\kappa'=2$) or Step \ref{kres_step} (otherwise).
\end{enumerate}
\end{enumerate}
\end{algorithm}

The steps of Algorithm \ref{QMBC_mitigation} are explained as follows. First, we circumvent the hardness of finding stopping sets \cite{McGregor, Krishnan} by running the BEC decoder, which fails on stopping sets. To focus on variable nodes that are likely to belong to a partially-erased stopping set, we rank the variable nodes according to their occurrences in the stopping sets found in Step \ref{step1}. We construct the subgraph induced by the union of the sets found in Step \ref{step1}, considered as a union of stopping sets, which is a stopping set as well. We then distribute resolvable edge labels using Theorem \ref{theorem:kappa} and Theorem \ref{th:2res}. Algorithm \ref{QMBC_mitigation} assumes no prior information on the code graph structure and requires no topology changes. Specifically, one of its advantages is that the degree distributions are not affected.

As an alternative to Algorithm \ref{QMBC_mitigation}, one may consider to distribute resolvable edge labels on the graph edges (i.e., without concentrating on stopping sets). However, this will result in a Tanner graph with at most ${s/{j_{\max }}} + 1$ edge label values instead of the possible $q-1=2^s-1$ edge labels. As a consequence, the probability of edge labels of the same value is increased, degrading the decoding performance (see Theorem \ref{thm:QMBC_BEC}). Thus, it is desired to first distribute the $q-1$ non-zero field elements uniformly on the edge labels and then to apply Algorithm \ref{QMBC_mitigation} to stopping sets only. The performance improvement of Algorithm \ref{QMBC_mitigation} is shown in Section~\ref{sec:sim_results}.


\section{Finite-Length Analysis of Maximum-Likelihood Decoding}
\label{sec:finite}

In this section, we analyze the finite-length decoding performance when a maximum-likelihood (ML) is used. We study the ML decoding performance for both the standard non-binary linear ensemble and LDPC ensembles. We denote by $\mathcal{E}_j$ ($j=1,2,...,s$) the index set of variable nodes partially-erased to $\mathcal{M}_0^j$ (see Section \ref{prelim}), and define $\mathcal{E} \buildrel \Delta \over = {\bigcup\limits_{j = 1}^s {{\mathcal{E}_j}} }$. We start with the following lemma.

\begin{lemma}
\label{lem:ML_decoding_all_zeros}
Consider a linear code used for transmission over the QMBC. The probability of decoding failure under ML decoding is independent of the transmitted codeword.
\end{lemma}

The proof of this lemma is provided in Appendix \ref{proof:ML_decoding_all_zeros}. As a result of Lemma \ref{lem:ML_decoding_all_zeros}, we assume in the rest of this section the transmission of the all-zero codeword. The following definition will serve us in analyzing the ML decoding performance.

\begin{definition}
\label{def:consistent}
A vector of length $\left| \mathcal{E} \right|$ with elements from GF($q$) is said to be \textit{consistent} with respect to $\left\{ {{{\cal E}_j}} \right\}_{j = 1}^s$ if an element of this vector indexed in $\mathcal{E}_j$ is contained in $\mathcal{M}_0^j$.
\end{definition}

\begin{example} \upshape
Suppose $q=4$, $\mathcal{E}_1 = \{ 1 \}$ and $\mathcal{E}_2 = \{2\}$, and consider the representation of GF($4$) as in Example \ref{ex1}. There are $8$ consistent vectors with respect to $\mathcal{E}_1, \mathcal{E}_2$: $\left( {0,0} \right)$, $\left( {0,1} \right)$,$\left( {0,\alpha } \right)$, $\left( {0,1 + \alpha } \right)$, $\left( {1,0} \right)$, $\left( {1,1} \right)$, $\left( {1,\alpha } \right)$ and $\left( {1,1 + \alpha } \right)$.
\end{example}


\subsection{Standard non-binary random ensemble}
\label{subsec:SNBRE}

In this part we calculate the expected probability of ML decoding failure over the standard non-binary random ensemble (SNBRE) of linear codes. Each code in the SNBRE is defined by a parity-check matrix $\bf{H}$ of dimensions $\left( {n - k} \right) \times n$, whose entries are i.i.d. uniform random variables taken from the GF($q$) elements. ${\bf{H}}_{{\cal{E}}}$ denotes its submatrix formed by the columns indexed in $\cal{E}$. To calculate the probability of decoding failure in the SNBRE case, we present the following definition.

\begin{definition}
\label{def_pr}
The columns of ${\bf{H}}_{\mathcal{E}}$ are said to be partially linearly independent if no consistent vector apart from the zero vector exists in the null space of ${\bf{H}}_{\mathcal{E}}$.
\end{definition}

The partial linear independence definition reduces to the ordinary linear independence definition when the partial erasures are full erasures (i.e., only $\mathcal{E}_s$ is non empty). However, the columns of $\bf{H}_\mathcal{E}$ can be partially linearly independent even if they are not linearly independent under the ordinary definition (e.g., when there are more columns than rows). This is demonstrated in the following example.

\begin{example} \upshape
Consider the representation of GF($4$) as in Example \ref{ex1}. Assume that $\left|\mathcal{E}_1\right| =2$ (all the other $\mathcal{E}_j$ are empty), such that the columns of $\bf{H}_{\mathcal{E}}$ are $\left(1,1\right)^T$ and $\left(\alpha,\alpha \right)^T$. These columns are linearly \textit{dependent} (e.g., the vector $\left(\alpha,1\right)^T$ is in the null space of $\bf{H}_{\mathcal{E}}$). However, there is no vector of length $2$ with elements taken from $\mathcal{M}_0^1 = \{0,1\}$ (with at least one non-zero element) in the null space of $\bf{H}_{\mathcal{E}}$. Therefore, the columns are partially linearly \textit{independent} according to Definition \ref{def_pr}.
\end{example}

To derive the probability of ML decoding success, we calculate the probability of partial linear independence. Let us define the set
\begin{equation}
\label{quot_set}
\mathcal{M}_0^{j,j'} \buildrel \Delta \over = {\left\{ {\frac{{{h_j}}}{{{h_{j'}}}}:{h_j} \in \mathcal{M}_0^j,{h_{j'}} \in \mathcal{M}_0^{j'}/\left\{ 0 \right\}} \right\}},
\end{equation}
obtained by an element-wise division of the set ${{\cal M}_0^j}$ by ${{\cal M}_0^{j'}/\left\{ 0 \right\}}$ (for certain $j,j' \le s$). Further, define $\chi ^{j,j'}$ as the cardinality of $\mathcal{M}_0^{j,j'}$:
\begin{equation}
\label{xi_def}
\chi ^{j,j'} \buildrel \Delta \over = \left|\mathcal{M}_0^{j,j'} \right|.
\end{equation}
Note that from group properties $\chi ^{j,j'}$ is symmetric, i.e., $\chi ^{j,j'}= \chi ^{j',j}$. In addition, $\chi ^{j,s} = q$ for any $j$.
\begin{example}
Assume that $q=4$. Then ${\chi ^{1,1}}=2$ and ${\chi ^{j,j'}}$ for $j \ne 1$ or $j' \ne 1$ are $4$.
\end{example}

Let $\psi$ denote the probability that the columns of a randomly drawn $\bf{H}_\mathcal{E}$ are partially linearly independent. For later use, we define ${x^ + } \buildrel \Delta \over = \max \left( {0,x} \right)$.

\begin{lemma}
\label{lem:psi_lb}
Given $\left\{ {{{\mathcal{E}_j}}} \right\}_{j = 1}^s$, let $\mathcal{O}$ contain all vectors of length $\left| {\cal E} \right|$ in which $j$ occurs $\left| {{\mathcal{E}_j}} \right|$ times. Then
\begin{align}
\label{psi_lb}
\psi \ge {\max _{\boldsymbol{o} \in \mathcal{O}}}\prod\limits_{i = 1}^{|\mathcal{E}|} {{{\left( {1 - \left( {\prod\limits_{l = 1}^{i - 1} {{\chi ^{{o_l},{o_i}}}} } \right)/{q^{n - k}}} \right)}^ + }}.
\end{align}
\end{lemma}

\begin{proof}
As the matrices in the SNBRE are equiprobable, $\psi$ is a function of $\left\{ {\left| {{\mathcal{E}_j}} \right|} \right\}_{j = 1}^s$ rather than of $\left\{ {{\mathcal{E}_j}} \right\}_{j = 1}^s$. Let us concentrate on some fixed but arbitrary choice of index sets with cardinalities $\left\{ {\left| {{{\cal E}_j}} \right|} \right\}_{j = 1}^s$. This choice is represented by a vector $\boldsymbol{o}$ that contains $j$ in indices of codeword symbols partially-erased to ${{\cal M}_0^j}$. Consider a matrix $\bf{H}_\mathcal{E}$ with columns $\boldsymbol{e}_i$ and denote by $\mathcal{A}_i$ the partial-erasure set indexed in $o_i$ ($i=1,2,...,\left|\mathcal{E}\right|$). We count in how many ways partially linearly independent columns can be placed in $\bf{H}_\mathcal{E}$.

Assume that the first $i' -1$ columns of $\bf{H}_{\mathcal{E}}$ are partially linearly independent. The next column, $\boldsymbol{e}_{i'}$, must satisfy ${\mathcal{A}_{i'}} \cdot {\boldsymbol{e}_{i'}} \ne \sum\limits_{l = 1}^{i'-1} {{\mathcal{A}_l} \cdot {\boldsymbol{e}_l}}$. Thus, $\boldsymbol{e}_{i'}$ must be different from the vectors in $\Gamma ={\sum\limits_{l = 1}^{i' - 1} {{\mathcal{A}_l}/\left\{ {{\mathcal{A}_{i'}} \setminus 0} \right\} \cdot {{\boldsymbol{e}}_l}} }$. The number of elements in $\Gamma$ is \textit{upper bounded} by $\prod\limits_{l = 1}^{i' - 1} {{\chi ^{{o_l},{o_{i'}}}}}$, as the linear combinations of $\boldsymbol{e}_l$ in $\Gamma$ might not be distinct. This is since linear independence in the ordinary sense is not necessarily guaranteed. We maximize over $\boldsymbol{o} \in \mathcal{O}$ to tighten the bound, and to obtain a probability we normalize by $q^{(n-k) |\mathcal{E}|}$, which is the number of possible $\bf{H}_\mathcal{E}$ matrices.
\end{proof}

Apart from the lower bound on $\psi$ of Lemma \ref{lem:psi_lb}, there are cases where the \textit{exact} value of $\psi$ can be found. Consider a subset $\mathcal{J}^*$ of $\left\{ {1,2,...,s} \right\}$ such that each element in $\mathcal{J}^*$ divides $s$ and $j'$ divides $j$ for all $j,j' \in \mathcal{J}^*$, $j' \le j$. We assume a representation of GF($q$) (see Section \ref{sec:gfq_rep}) such that for each $j^* \in \mathcal{J}^*$, the partial erasure set $\mathcal{M}_0^{j^*}$ is mapped to a \textit{subfield} of GF($q$) (i.e., in addition to being an additive subgroup of GF($q$)). Moreover, for each pair $j,j' \in \mathcal{J}^*$, $j' \le j$, $\mathcal{M}_0^{j'}$ is mapped to a subfield of $\mathcal{M}_0^{j}$.

\begin{example}
If $q=4$, the possible choices of $\mathcal{J}^*$ are $\{1\}$, $\{2\}$ and $\{1,2\}$. If $q=8$, $\mathcal{J}^*$ can be $\{1\}$, $\{2\}$, $\{1,2\}$ or $\{1,3\}$.
\end{example}
The following lemma shows that when the divisibility conditions above are met, the upper bound on $\psi$ via sequential exclusion of dependencies (Lemma~\ref{lem:psi_lb}) becomes exact, {\em if we sort the partial erasures in non-increasing order}.
\begin{lemma}
\label{lem:psi_subfield}
Assume that $\mathcal{E}_{{j}} = \emptyset$ for $j \notin \mathcal{J}^*$. Denote by $\boldsymbol{o}$ the (now specific) vector of length $|\mathcal{E}|$ with $s$ in its first $|\mathcal{E}_{s}|$ entries, $s-1$ in its next $|\mathcal{E}_{s-1}|$ entries downto $1$ in its last $|\mathcal{E}_{1}|$ entries. Then
\begin{align}
\psi = \prod\limits_{i = 1}^{|\mathcal{E}|} {{{\left( {1 - \left( {\prod\limits_{l = 1}^{i - 1} {2^{{o}_l}} } \right)/{q^{n - k}}} \right)}^ + }} .
\end{align}
\end{lemma}

\begin{proof}
\label{proof:lem:psi_subfield}

Consider the placement process depicted in the proof of Lemma \ref{lem:psi_lb} and assume that we place the  $i'$th column. From the  ordering of ${\boldsymbol{o}}$ we get that $\chi^{{o_l},{o_{i'}}}=2^{o_l}$ for $l<i'$. In choosing the vector $\boldsymbol{e}_{i'}$ we exclude all combinations of previous vectors $\boldsymbol{e}_l$ with coefficients in $\mathcal{M}_0^{{o_l},{o_{i'}}}$. Assume by contradiction that two of these $\prod\limits_{l = 1}^{i' - 1} {2^{o_l}}$ combinations result in the same vector. But this would imply an $\boldsymbol{e}_{i''}$, $i''<i'$, that is a combination of vectors $\boldsymbol{e}_l$, $l<i''$, with coefficients in $\mathcal{M}_0^{{o_l},{o_{i'}}}$. Since for any $l$ , $\mathcal{M}_0^{{o_l},{o_{i'}}}=\mathcal{M}_0^{{o_l},{o_{i''}}}$, this is a contradiction because it means that at step $i''$ we did not exclude all partially dependent vectors, and thus the count is exact with no over-subtraction.
%
%
\end{proof}
Based on either the lower bound on $\psi$ of Lemma \ref{lem:psi_lb} or its exact value for the cases of Lemma \ref{lem:psi_subfield}, we calculate the expected value of $P_{{\rm{error}}}^{{\rm{ML}}}$ for the SNBRE.

\begin{theorem}
\label{thm:PML_B}

The expected probability of decoding failure over codes drawn from the non-binary random ensemble under ML decoding is
\begin{align}
\label{PML_B}
&\mathbb{E}_{{\rm SNBRE}}\left[ P^{{\rm{ML}}}\left(\bf{H}\right) \right]
\\\nonumber
&\le \sum\limits_{\scriptstyle\left| {{\mathcal{E}_0}} \right|,\left| {{\mathcal{E}_1}} \right|,...,\left| {{\mathcal{E}_s}} \right|:\hfill\atop \scriptstyle \sum\limits_{j = 0}^s {\left| {{\mathcal{E}_j}} \right|}  = n\hfill} {\frac{{n!}}{{\left| {{\mathcal{E}_0}} \right|!\left| {{\mathcal{E}_1}} \right|!...\left| {{\mathcal{E}_s}} \right|!}}\prod\limits_{j = 0}^s {{{ {{\varepsilon_j}} }^{{|\mathcal{E}_j|}}}} } \cdot \left(1 - \tilde{\psi} \right),
\end{align}
where $\tilde{\psi}$ is\footnote{While implicit in the expressions, recall that $\tilde{\psi}$ depends on $\left\{ {\left| {{\mathcal{E}_j}} \right|} \right\}_{j = 1}^s$.} either the lower bound of Lemma \ref{lem:psi_lb}, or its exact value in the cases of Lemma \ref{lem:psi_subfield} (in the latter cases, an equality is attained in \eqref{PML_B}).
\end{theorem}

\begin{proof}
Recall that the transmission of the all-zero codeword is assumed without loss of generality. Consider a fixed but arbitrary partial-erasure index sets $\left\{ {{\mathcal{E}_j}} \right\}_{j = 1}^s$. The channel output is not resolved as the all-zero codeword if and only if there is a non-zero consistent solution to ${{\bf{H}}_{\cal E}} \boldsymbol{x}_\mathcal{E}^T=\boldsymbol{0}$. This happens if the columns of $\bf{H}_{\mathcal{E}}$ are partially linearly independent, with probability which is $1-\psi$. Finally, we sum over the possible cardinalities of the partial-erasure index sets, using the multinomial distribution and the channel partial-erasure probabilities, to obtain \eqref{PML_B}.
\end{proof}


If $s=1$ and all the partial-erasure sets are $\{0,1\}$ (i.e., BEC full-erasures), we obtain \cite[Theorem 3.1]{Di} as a special case of Theorem \ref{thm:PML_B} (with equality). In Figure \ref{fig:esnbre} we plot $\mathbb{E}_{{\rm SNBRE}}\left[ P^{{\rm{ML}}}\left(\bf{H}\right) \right]$ for a $q=4$ channel with $\varepsilon_2 = \varepsilon_1/10$ and different $n$ values. This is compared to an asymptotically equivalent $q$-ary erasure channel (QEC), i.e., with $\varepsilon = \varepsilon_1/2 + \varepsilon_2$. It is demonstrated that the QMBC finite-length ML performance is orders of magnitude better, though the Shannon limit is the same.

\begin{figure}[!t]
\centering
\includegraphics[scale=0.6]{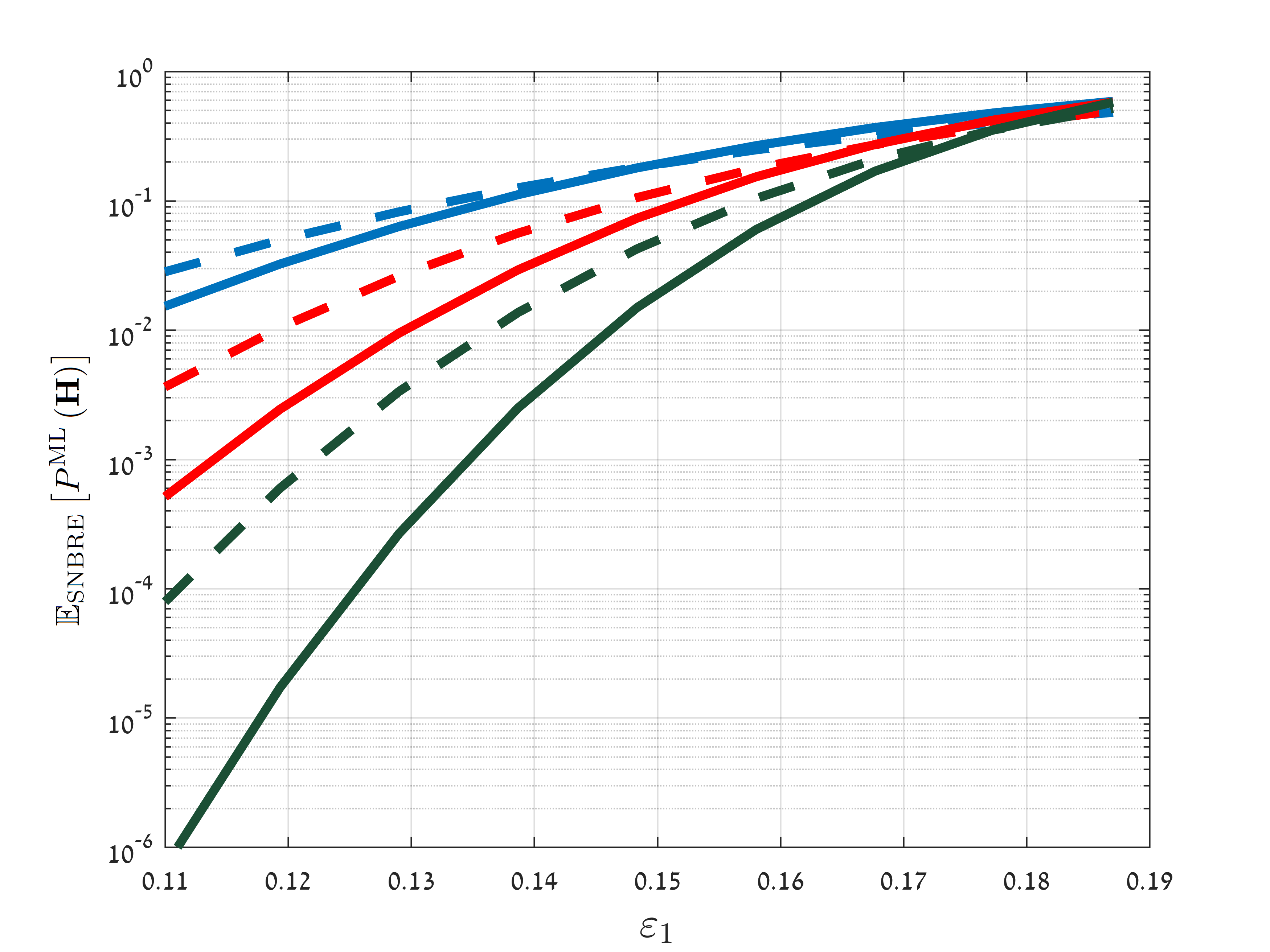}
\caption{Exact $\mathbb{E}_{{\rm SNBRE}}\left[ P^{{\rm{ML}}}\left(\bf{H}\right) \right]$ as a function of $\varepsilon_1$, for $\varepsilon_2 = \varepsilon_1/10$ and $q=4$ (solid lines). An asymptotically equivalent QEC with $\varepsilon=\left(3/5\right) \varepsilon_1$ is also shown (dashed lines). The codeword lengths are $n=128,256,512$ (top to bottom) and the rate is $8/9$ (Shannon limit: $0.185$).}
\label{fig:esnbre}
\end{figure}

\subsection{LDPC ensembles under ML decoding}

In this part, we derive an upper bound on the expected ML decoding performance over the regular non-binary ($d_v,d_c$) LDPC ensemble. We start with the following lemma, which will serve us later in calculating the probability that a certain check node is satisfied.

\begin{lemma}
\label{lem:zero_sum}
Consider a vector $\boldsymbol{a}$ of length $m \ge 2$, whose entries are i.i.d. random variables  uniformly distributed on the non-zero GF($q=2^s$) elements. The probability that the entries of $\boldsymbol{a}$ sum to $0$ is
\begin{align}
\label{w_1}
\Pr \left( {\sum\limits_{i = 1}^m {{a_i}}  = 0} \right) &= \frac{{1 - {{\left( {1 - q} \right)}^{1 - m}}}}{q} \le \frac{1}{q-1}.
\end{align}
\end{lemma}

\begin{figure}[!t]
\centering
\includegraphics[scale=0.6]{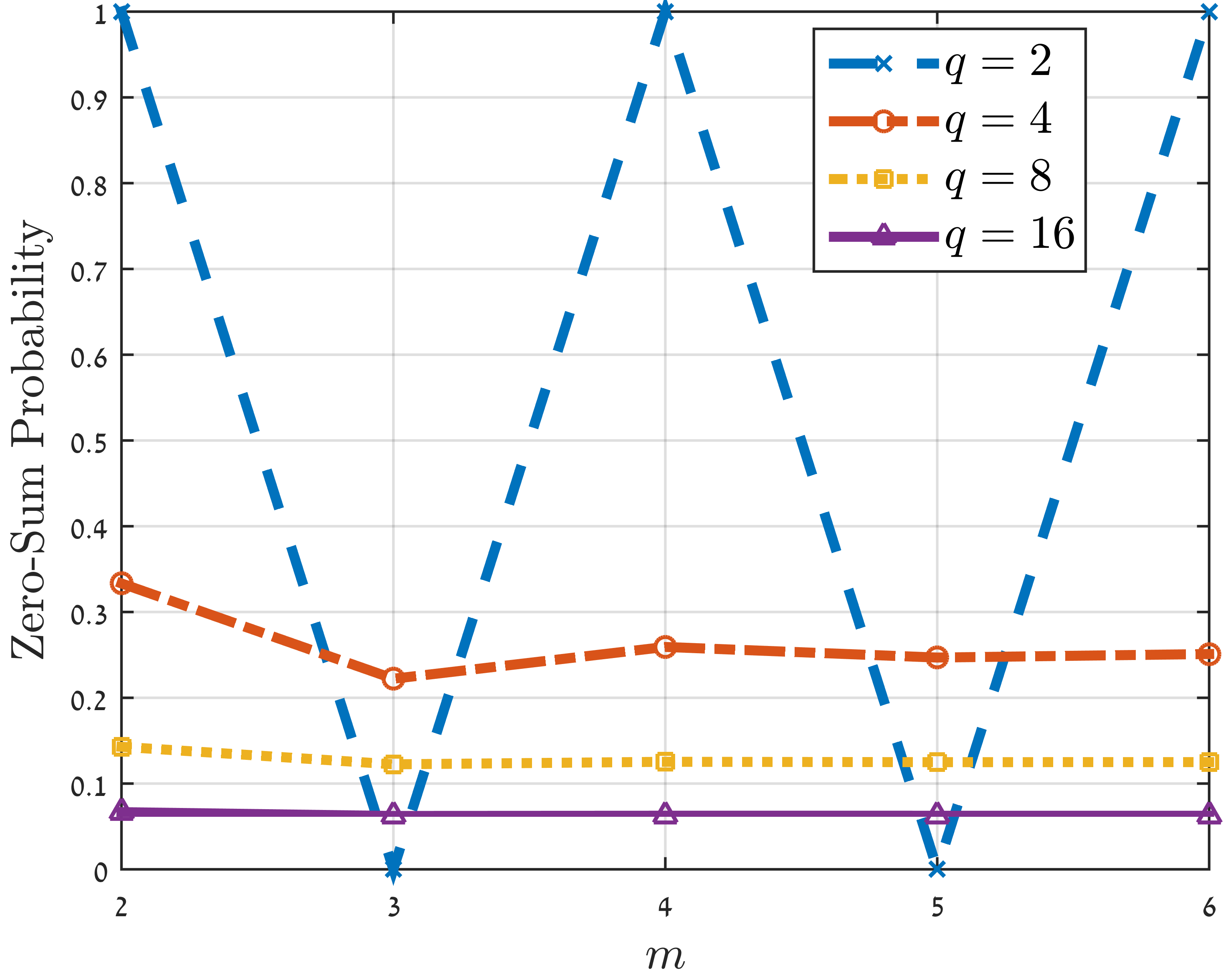}
\caption{The probability that a check node is satisfied given $m$ non-zeros among its connected variable nodes, under the uniform distribution of the edge labels. The binary case ($q=2$, no sensitivity to edge labels) is provided for reference.}
\label{fig:zero_sum_prob}
\end{figure}

The proof of this lemma is provided in Appendix \ref{proof:zero_sum}. In Figure \ref{fig:zero_sum_prob}, the zero-sum probability $\Pr \left( {\sum\limits_{i = 1}^m {{a_i}}  = 0} \right)$ is shown for several values of $m$ and $q$. Note that in the binary case ($q=2$) this probability is $1$ if $m$ is even, and $0$ otherwise, as expected. It is demonstrated in Figure \ref{fig:zero_sum_prob} that the zero-sum probability is approximately independent of $m$ when $q \ge 2$, and that the upper bound ${1}/({q-1})$ is tight. Note that $\Pr \left( {\sum\limits_{i = 1}^m {{a_i}}  = 0} \right)$ depends on the number of non-zero entries in $\boldsymbol{a}$ and not on the entries themselves. In the following lemma, we calculate the number of consistent vectors (see Definition \ref{def:consistent}) with a certain number of non-zero entries.

\begin{lemma}
\label{lem:weight}
Given $\mathcal{E}=\left\{ {{{\mathcal{E}_j}}} \right\}_{j = 1}^s$, the number of vectors with $w$ non-zero entries that are consistent with $\mathcal{E}$ is
\begin{equation}
\eta \left( w \right) = \sum\limits_{\scriptstyle{\boldsymbol{u}}:\sum\limits_{j = 1}^s {{u_j}}  = w,\hfill\atop
\scriptstyle{u_j} \le \left| {{{\cal E}_j}} \right|\hfill} {\prod\limits_{j = 1}^s {\left( \begin{array}{c}
\left| {{{\cal E}_j}} \right|\\
{u_j}
\end{array} \right)} } {\left( {{2^j} - 1} \right)^{{u_j}}}.
\end{equation}
\end{lemma}
\begin{proof}
An element $u_j$ of $\boldsymbol{u}$ counts the number of non-zero entries taken from the $\left| {{\mathcal{E}_j}} \right|$ partial-erasure set $\mathcal{M}_0^j$. The number of ways to choose the locations of the partial-erasure sets is counted with the factor ${\left| {\cal E}_j \right| \choose {u_j}}$, where for each choice there are ${\left( {{2^j} - 1} \right)^{{u_j}}}$ ways to choose the non-zero entries.
\end{proof}

Note that when $s=1$ and all the partial-erasure sets are $\{0,1\}$ (i.e., BEC full-erasures), $\eta \left( w \right)$ degenerates into ${\left| {\cal E} \right| \choose {w}}$, which is the number of binary vectors of length $|\mathcal{E}|$ whose Hamming weight is $w$. Let us denote by $P^{{\rm{ML}}}\left(\mathcal{G}\right)$ the probability of ML decoding failure for a certain Tanner graph $\mathcal{G}$ from the regular $\left(d_v,d_c\right)$ ensemble. We now use Lemma \ref{lem:zero_sum} and Lemma \ref{lem:weight} to upper bound the expected value (over graphs in the $\left(d_v,d_c\right)$ ensemble) of $P^{{\rm{ML}}}\left(\mathcal{G}\right)$. As in \cite{Di, Orlitsky}, we use polynomial characteristic functions to identify graph configurations leading to failure events. We denote by ${\rm{coef}}\left( {f\left( x \right),{x^i}} \right)$ the $i\rm{th}$ coefficient $f_i$ of $x^i$ in the polynomial $f\left( x \right) = \sum\limits_{i \ge 0} {{f_i}{x^i}}$ (note that ${\rm{coef}}\left( {{{\left( {1 + y} \right)}^n},{x^k}} \right) = {n \choose k}$). We also denote by $\mathbb{E}_{{\rm LDPC} (d_v,d_c)}\left[ P^{{\rm{ML}}}\left(\mathcal{G}\right) \right]$ the expected probability of decoding failure, where the expectation is taken over LDPC codes in the $(d_v,d_c)$ ensemble. Recall that $\eta \left( w \right)$ is a function of $\left\{ {\left| {{\mathcal{E}_j}} \right|} \right\}_{j = 1}^s$.

\begin{theorem}
\label{thm:ub_LDPC}
\begin{align}
\label{LDPC_ML}
&\mathbb{E}_{{\rm LDPC} (d_v,d_c)}\left[ P^{{\rm{ML}}}\left(\mathcal{G}\right) \right]
\le
\\\nonumber
&\sum\limits_{\scriptstyle\left| {{\mathcal{E}_0}} \right|,\left| {{\mathcal{E}_1}} \right|,...,\left| {{\mathcal{E}_s}} \right|:\hfill\atop \scriptstyle \sum\limits_{j = 0}^s {\left| {{\mathcal{E}_j}} \right|}  = n\hfill} {\frac{{n!}}{{\left| {{\mathcal{E}_0}} \right|!\left| {{\mathcal{E}_1}} \right|!...\left| {{\mathcal{E}_s}} \right|!}}\prod\limits_{j = 0}^s {{{ {{\varepsilon_j}}}^{{|\mathcal{E}_j|}}}} }
\\\nonumber
&\cdot \min \bigg\{ {1,\sum\limits_{w = 1}^{|\mathcal{E}|} {{\eta(w)} } \frac{{{\rm{coef}}\left( {{{\left( {{{\left( {1 + y} \right)}^{d_c}} - 1 - yd_c} \right)}^{n\frac{{{d_v}}}{{{d_c}}}}},{y^{wd_v}}} \right)}}{{ {nd_v \choose wd_v} }}}
\\\nonumber
&{{\left( {\frac{1}{{q - 1}}} \right)}^{w\frac{{{d_v}}}{{{d_c}}}}}\bigg\}.
\end{align}


\end{theorem}
\begin{proof}
An ML decoder fails if and only if there is a non-trivial solution to the equation ${{{\bf{H}}_\mathcal{E}}{\boldsymbol{x}}_{\cal E}^T = {\boldsymbol{0}}}$, which is consistent with respect to $\left\{ {{{\cal E}_j}} \right\}_{j = 1}^s$:
\begin{align}
\label{prob_LDPC}
&\Pr \left( {\exists \boldsymbol{x}_\mathcal{E} \ne \boldsymbol{0}, \boldsymbol{x}_\mathcal{E} {\hspace{2pt} \rm is \hspace{2pt} consistent}:{\bf{H}_\mathcal{E}}{\boldsymbol{x}_\mathcal{E}^T} = \boldsymbol{0}} \right)
\\\nonumber
& \le \sum\limits_{{\boldsymbol{x}_\mathcal{E}} \ne \boldsymbol{0}, \boldsymbol{x}_\mathcal{E} \hspace{2pt} \rm is \hspace{2pt} consistent } {\Pr \left( {{{\bf{H}}_{\cal E}}{{\boldsymbol{x}_\mathcal{E}^T}} = \boldsymbol{0}} \right)},
\end{align}
where the upper bound follows by the union bound. Consider an arbitrary but fixed consistent vector $\boldsymbol{x}_\mathcal{E}$ and denote the number of its non-zero entries by $w(\boldsymbol{x}_\mathcal{E})$. There are $w(\boldsymbol{x}_\mathcal{E}) d_v $ edges connected to variable nodes corresponding to the non-zero elements of $\boldsymbol{x}_\mathcal{E}$. For  ${{{\bf{H}}_{\cal E}}{\boldsymbol{x}}_{\cal E}^T = {\boldsymbol{0}}}$ to hold, each neighbouring check of the $w(\boldsymbol{x}_\mathcal{E})$ non-zero variable nodes must be connected to these variable nodes at least twice. As the total number of check nodes is $n d_v /d_c$, we have ${{\rm{coef}}\left( {{{\left( {{{\left( {1 + y} \right)}^{d_c}} - 1 - d_cy} \right)}^{n\frac{{{d_v}}}{{{d_c}}}}},{y^{w(\boldsymbol{x}_\mathcal{E})d_v}}} \right)}$ configurations out of $nd_v \choose w(\boldsymbol{x}_\mathcal{E})d_v$ such configuration. According to Lemma \ref{lem:zero_sum}, the probability that a certain check node is satisfied is upper bounded by $1/(q-1)$ (recall that uniform edge labels are assumed). The number of check nodes connected to $w(\boldsymbol{x}_\mathcal{E})$ variable nodes is at least $w(\boldsymbol{x}_\mathcal{E}) d_v /d_c $.  Thus, ${\left( {1/\left( {q - 1} \right)} \right)^{w(\boldsymbol{x}_\mathcal{E}){d_v}/{d_c}}}$ is an upper bound on the probability that all check nodes connected to the $w(\boldsymbol{x}_\mathcal{E})$ non-zero variable nodes are satisfied. Finally, by summing over all the possible weights of consistent vectors (counted by $\eta\left(w\right)$ of Lemma \ref{lem:weight}) and taking into account the channel partial-erasure probabilities, \eqref{LDPC_ML} is obtained. The minimum in \eqref{LDPC_ML} is taken to tighten the upper bound.
\end{proof}

In Figure \ref{fig:LDPC_comparison}, we compare \eqref{LDPC_ML} for $q=4$, where the set $\{0,1\}$ is considered as either a partial erasure (decoded with the QMBC decoder) or a full erasure (decoded with the BEC decoder). In terms of the upper bound \eqref{LDPC_ML}, the QMBC model is expected to provide ML decoding performance orders of magnitude better compared to full-erasure decoding.

\begin{figure}[!t]
\centering
\includegraphics[scale=0.6]{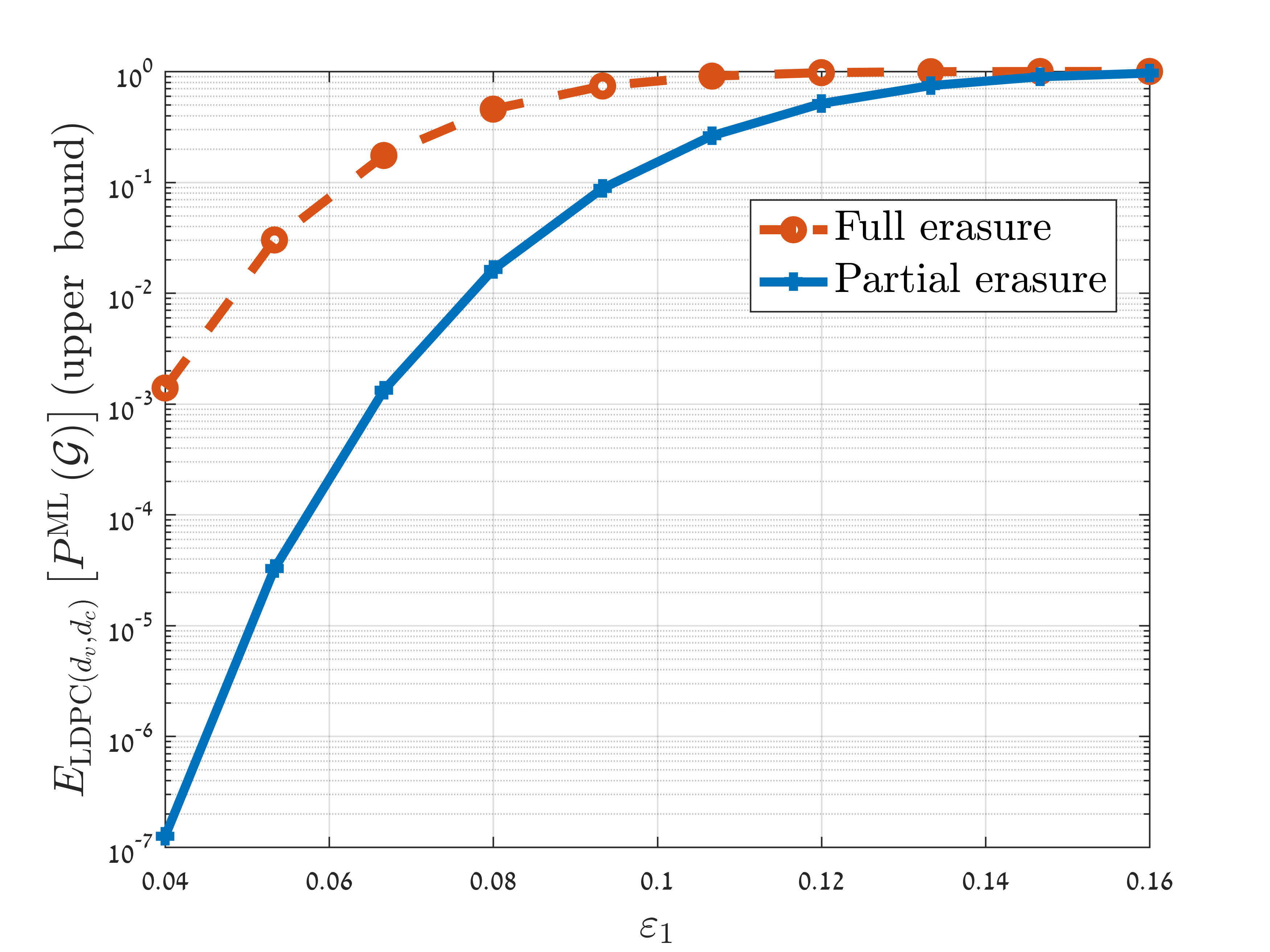}
\caption{A comparison of $\mathbb{E}_{{\rm LDPC} (d_v,d_c)}\left[ P^{{\rm{ML}}}\left(\mathcal{G}\right) \right]$ for the LDPC ensemble $(3,27)$ (rate $8/9$), for a GF($4$) code of length $252$. The set $\{0,1\}$ is either considered as a partial erasure or a full erasure with probability $\varepsilon_1$.}
\label{fig:LDPC_comparison}
\end{figure}

\section{Simulation Results}
\label{sec:sim_results}

In this part, we present simulation results of the QMBC iterative-decoding performance. We used the regular ($3,27$) LDPC code ensemble (rate $8/9$), whose rate is of interest in practical flash memories. Two codeword lengths were considered: $n=513$ and $n=1026$. The average decoding performance is measured by symbol erasure rate (SER), where each variable node that remains partially erased when the decoder terminates contributes to this quantity.

\subsection{Comparison to binary full erasures}

As a preliminary step, we considered binary coding with GF($q$) symbols converted to bits. In this setting, a GF($q$) symbol is decomposed into $s$ bits, where a partial-erasure event of type $j$ corresponds to $j$ (fully) erased least significant bits. We compare GF($q$) codes with partial erasures (decoded using the QMBC decoder) to binary codes with equivalent full erasures (decoded using the BEC decoder). The results are shown in Figure \ref{fig:sim_results_binary}. It is demonstrated that partial-erasure decoding outperforms binary erasure decoding, offering SER performance better by up to an order of magnitude. The improved performance of GF($q$) codes over binary codes is explained by the mitigated effect of stopping sets due to the non-binary edge labels, as we developed in Section \ref{sec:ss}.


\begin{figure}[h!]
        \centering
        \begin{subfigure}[t]{0.5\textwidth}
                \includegraphics[scale=0.6]{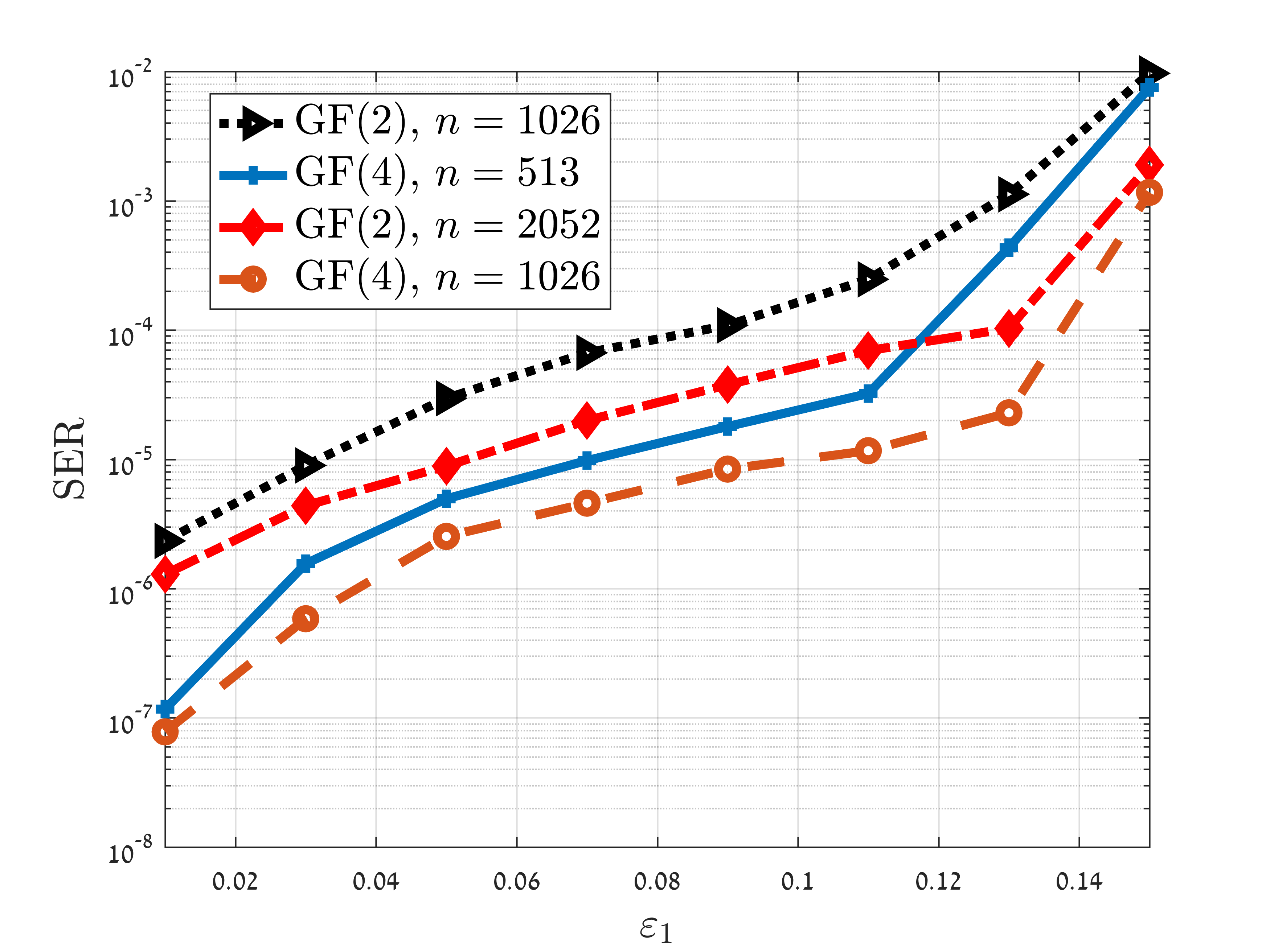}
                \caption{$q=4$, $j=1$ (decoding threshold: $0.184$).}
                \label{fig:q4_2}
        \end{subfigure}%
        ~ 

        \begin{subfigure}[t]{0.5\textwidth}
                \includegraphics[scale=0.6]{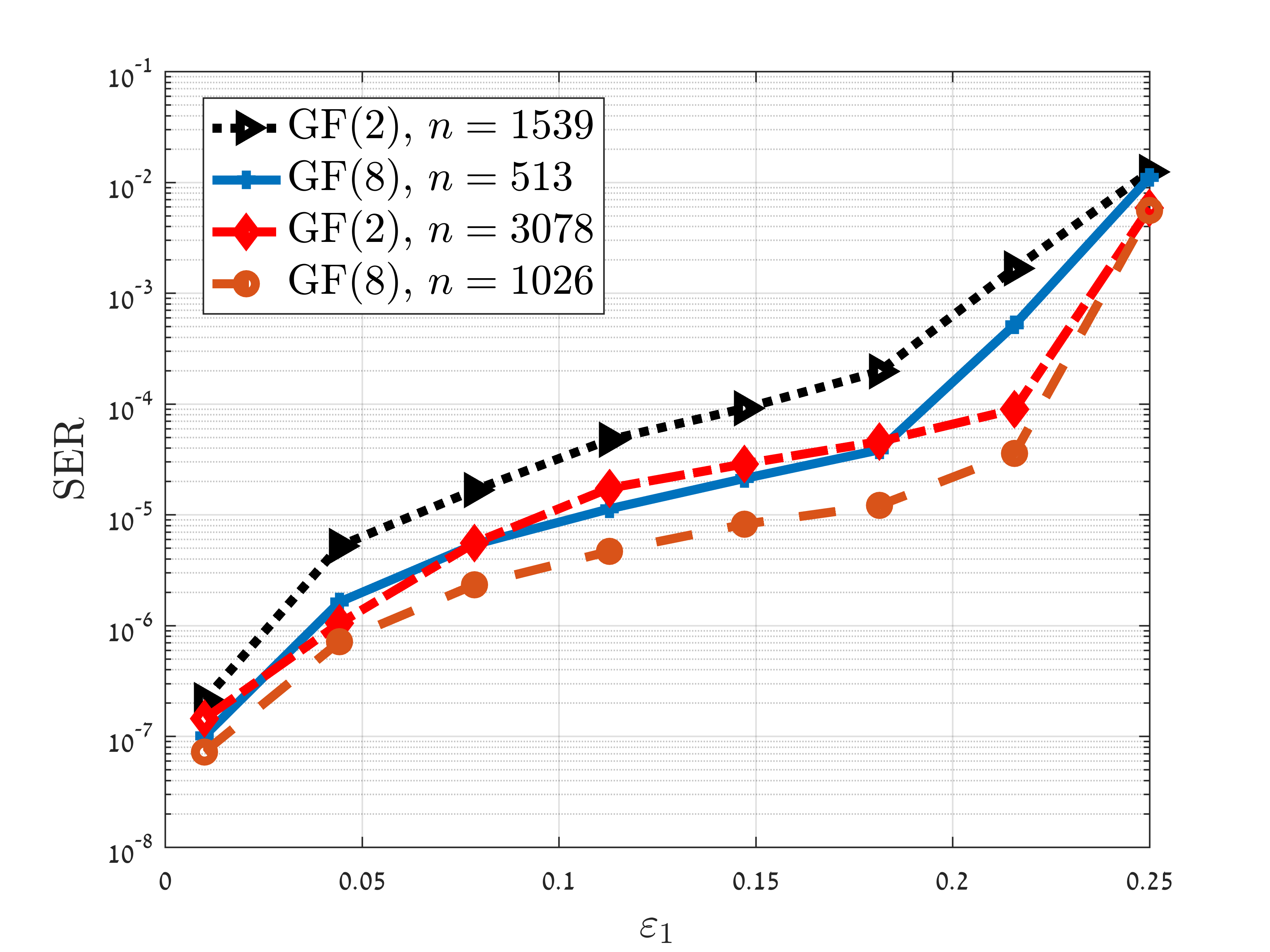}
                \caption{$q=8$, $j=1$ (decoding threshold: $0.276$).}
                \label{fig:q8_2}
        \end{subfigure}

        ~ 
\caption{SER performance comparison between GF($q$) and binary codes. The labels of the GF($q$) LDPC codes are uniformly distributed.}
\label{fig:sim_results_binary}
\end{figure}

\subsection{Performance of the edge-labeling algorithm}

As we saw in the previous subsection, GF($q$) LDPC codes over the QMBC are superior to binary codes. In this part, we show that the decoding performance of GF($q$) LDPC codes can be further improved using the edge-labeling algorithm (Algorithm \ref{QMBC_mitigation}) developed in Section \ref{ss:edge_labeling}. In Figure \ref{fig:sim_results}, we compare the iterative decoding performance of GF($q$) with uniformly-distributed edge labels to edge labels optimized using Algorithm \ref{QMBC_mitigation}. The optimized edge labels lead to a significant improvement in in SER performance, up to two orders of magnitude. It is demonstrated that the performance gap increases with $q$ for a fixed partial-erasure type. The reason is the larger number of resolvable edge labels, which increases with $q$ (see Theorem \ref{theorem:kappa}).

\begin{figure}[h!]
        \centering
        \begin{subfigure}[t]{0.5\textwidth}
                \includegraphics[scale=0.6]{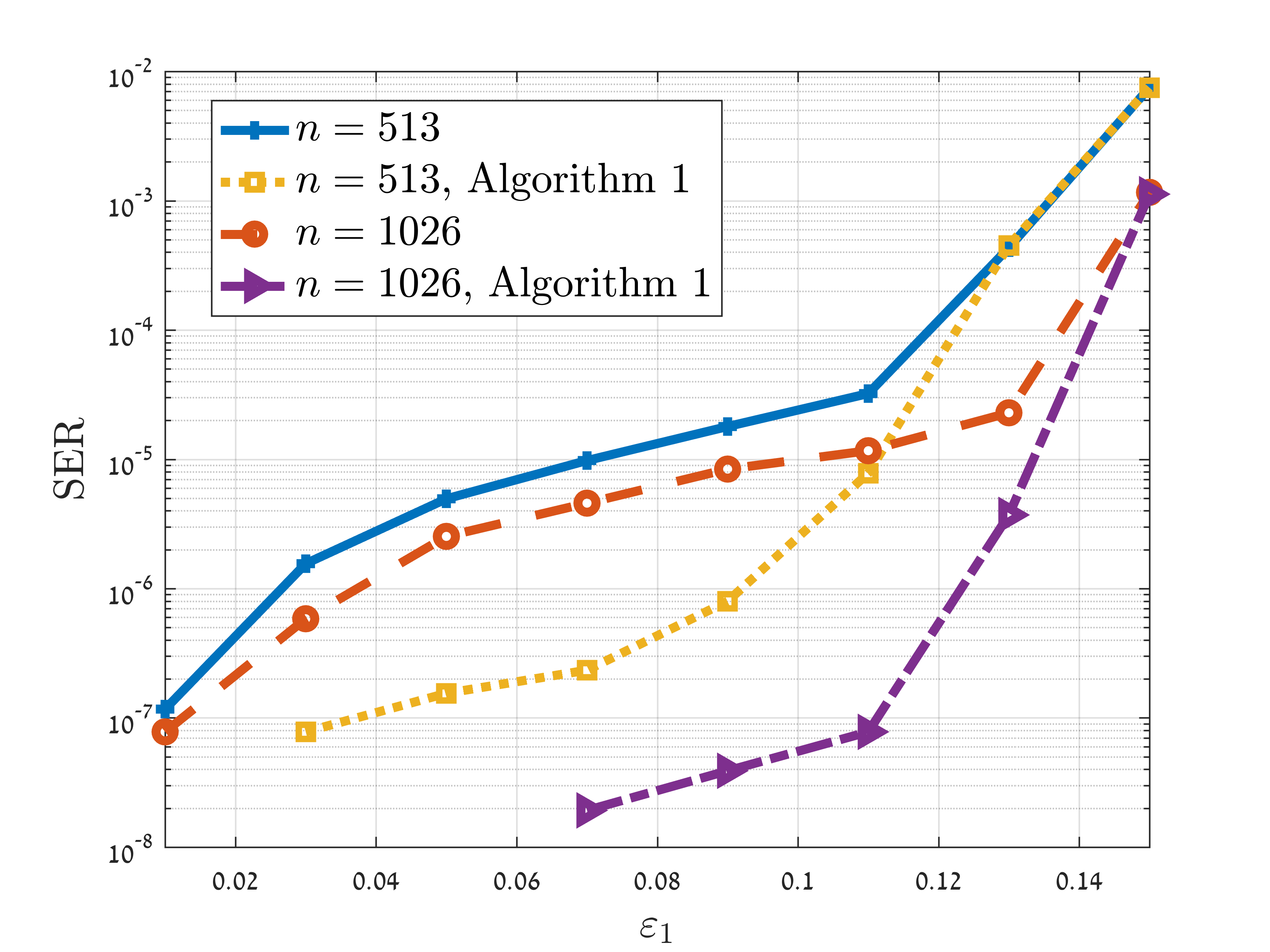}
                \caption{$q=4$, $j=1$ (decoding threshold $0.184$).}
                \label{fig:q4}
        \end{subfigure}%
        ~ 

        \begin{subfigure}[t]{0.5\textwidth}
                \includegraphics[scale=0.6]{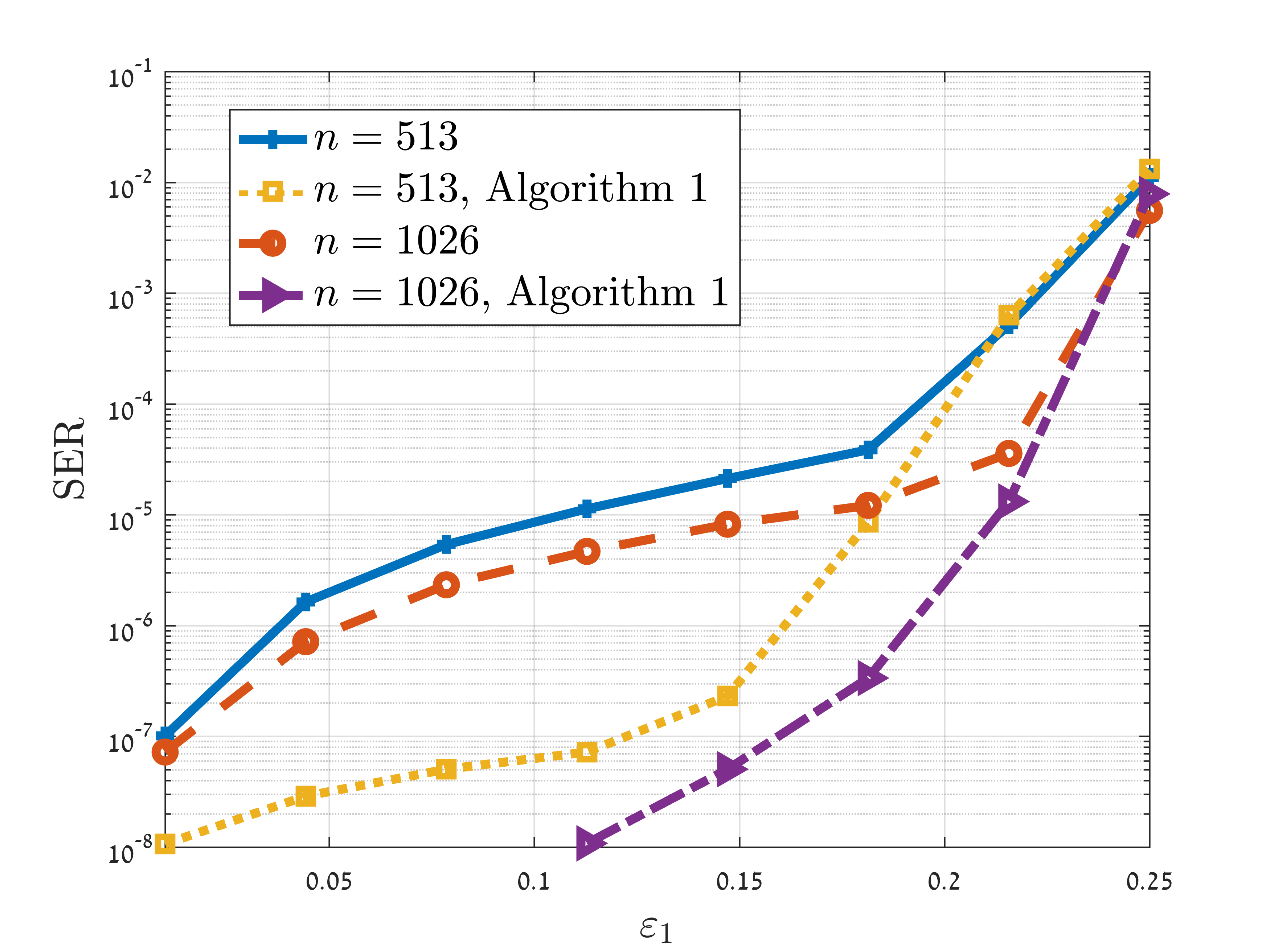}
                \caption{$q=8$, $j=1$ (decoding threshold $0.276$).}
                \label{fig:q8}
        \end{subfigure}

%
        ~ 
\caption{SER performance comparison of QMBC partial-erasure decoding, between uniformly-distributed and optimized edge labels. The decoding thresholds are given for optimal edge-label distributions.}
\label{fig:sim_results}
\end{figure}

\section{Conclusion}
\label{sec:conclusion}

This work offers a study of the performance of iterative decoding of GF($q$) LDPC codes over the QMBC. By an asymptotic threshold analysis, we demonstrated explicitly how the edge label distribution affects decoding performance. We later showed that unlike the binary case, partially-erased stopping sets can be resolved by a wise setting of edge labels. For this aim, we proposed and evaluated an edge-labeling algorithm for improved finite-length decoding performance. Finally, we derived expressions for the finite-length performance of a maximum-likelihood decoder, both for the standard non-binary random ensemble and for LDPC ensembles.

Our work leaves interesting problems for future research. Designing good GF($q$) LDPC codes for the QMBC is an important research direction. Unlike binary codes, GF($q$) LDPC codes require a joint optimization of degree and edge-label distributions. It is of importance to give an expression for the QMBC finite-length performance that depends on the edge-label distribution in addition to the stopping-set distribution. As another direction, the upper bound on the ML decoding performance for LDPC ensembles might be improved by considering non-uniform edge-label distributions.

\bibliographystyle{IEEEtran}
	\bibliography{QMBC}
\appendices

\section{Proof of Theorem \ref{th:capacity}}
\label{proof:capacity}

Define ${p_x} \triangleq \Pr (X = x)$ for $x \in \mathcal{X}$ to be the input distribution to the channel. The channel capacity $C$ is:
\begin{equation}
\label{cap_def}
C = {\max _{\left\{ {{p_x}} \right\}}}I\left( {X;Y} \right) = {\max _{\left\{ {{p_x}} \right\}}}\left( {H\left( Y \right) - H\left( {\left. Y \right|X} \right)} \right),
\end{equation}
where $I\left( {X;Y} \right)$ is the mutual information between the input $X$ and the output $Y$, and $H\left( Y \right)$, ${H\left( {\left. Y \right|X} \right)}$ are the entropy of $Y$ and the conditional entropy of $Y$ given $X$, respectively. The conditional entropy of $Y$ given $X$ can be calculated using the transition probabilities in \eqref{tran_matrix}. $Y$ is a set of $2^j$ elements with probability $\varepsilon_j$. Thus,
\begin{equation}
\label{HYX}
H\left( {\left. Y \right|X} \right) = -\sum\limits_{j = 0}^s {{\varepsilon _j}\log } \left( {{\varepsilon _j}} \right).
\end{equation}
We now move to maximize $H(Y)$, since $H(Y|X)$ is independent of the input distribution. For a given $j$, let us denote by $\Omega _i^j$ ($i=1,2,...,q/2^j$) the \textit{distinct} sets among $\mathcal{M}_x^j$. The entropy $H(Y)$ as a function of ${\left\{ {{p_x}} \right\}}$ is:
\begin{align}
\label{HY}
H\left( Y \right) =  -\sum\limits_{j = 0}^s {\sum\limits_{i = 1}^{q/{2^j}} {\left( {\sum\limits_{x \in \Omega _i^j} {{p_x}} {\varepsilon _j}} \right)\log \left( {\sum\limits_{x \in \Omega _i^j} {{p_x}} {\varepsilon _j}} \right)} }.
\end{align}

The capacity-achieving distribution can be found by solving the following maximization problem:
\begin{equation}
{\max _{\left\{ {{p_x}} \right\}}} H(Y),\hspace{7pt} {\text{s}}.{\text{t}}.\sum\limits_{x \in \mathcal{X}} {{p_x}}  = 1.
\end{equation}
Define the function $f\left( {\left\{ {{p_x}} \right\}_{x \in \mathcal{X}}} \right)$ to be the entropy $H(Y)$ as a function of $p_x$. Using the method of Lagrange multipliers, we get the following system of equations:
\begin{equation}
\label{lagrange}
\frac{{\partial f}}{{\partial {p_x}}} + \lambda  = 0,\hspace{7pt} \sum\limits_{x \in \mathcal{X}} {{p_x}}  = 1,
\end{equation}
where $\lambda$ is the Lagrange multiplier. The left-hand side of Equation \eqref{lagrange} leads to the following equations for $x \in \mathcal{X}$:
\begin{align}
 - \sum\limits_{j = 0}^s {\sum\limits_{i = 1}^{q//{2^j}} {\left[ {{\varepsilon _j}\log \left( {\sum\limits_{x \in \Omega _i^j} {{p_x}} {\varepsilon _j}} \right) + \sum\limits_{x \in \Omega _i^j} {{\varepsilon _j}} } \right]} }  + \lambda  = 0.
\end{align}
The equations are satisfied for the uniform distribution $p_x = 1/q$, where $\lambda$ assumes a constant value independent of $x$. As $I\left( {X;Y} \right)$ is a concave function of $p_x$ once $\Pr \left( {\left. {Y = y} \right|X = x} \right)$ is given, the uniform distribution leads to the global maximum of $I\left( {X;Y} \right)$, that is, to the capacity. Finally, calculating the capacity using Equations \eqref{HYX}-\eqref{HY} with $p_x$ substituted by $1/q$, leads to the capacity in \eqref{qmbc_capacity}. \qed

\section{Proof of Theorem \ref{thm:QMBC_BEC}}
\label{proof:QMBC_BEC}

Assume that all the edge labels are the same. In this case, the CTV messages are independent of the edge labels, and we have (see \eqref{CTV_def}):
\begin{equation}
{\rm CTV}_{\mathtt{c} \to \mathtt{v}}^{\left( l \right)} =  \sum\limits_{\mathtt{v}' \in \left\{ {\mathcal{N}\left( \mathtt{c} \right)\backslash \mathtt{v}} \right\}} { {{\rm VTC}_{\mathtt{v}' \to \mathtt{c}}^{\left( l-1 \right)}}}.
\end{equation}
That is, an outgoing CTV message is simply the sumset of the incoming VTC messages. Recall that the initial channel-information sets are contained in each other, i.e. $\mathcal{M}_0^j \subseteq \mathcal{M}_0^{j'}$ for $j \le j'$, and that each set is an additive subgroup of $\text{GF}^{+}$($q$), closed under addition. For example, the possible channel-information sets when $q=4$ are $\left\{ 0 \right\},\left\{ {0,1} \right\}$ and $\left\{ {0,1,2,3} \right\}$ (we define ${\mathcal{M}_0^0}$ as the singleton $\left\{ 0 \right\}$). Due to the closure property of subgroups, the initial sumset at a check node can be written as:
\begin{equation}
\label{sumset_union}
\sum\limits_{j \in \mathcal{M}_{\mathtt v}} {\mathcal{M}_0^j}  = \mathcal{M}_0^{\mathop {\max}\limits_{j \in \mathcal{M}_{\mathtt v}} j},
\end{equation}
where $\mathcal{M}_{\mathtt{v}}$ is an ordered list containing indices of incoming VTC messages (See Section \ref{sec:de_equations}). Thus, the sumset operation at check nodes simplifies to finding the incoming VTC message of the maximum cardinality. In a similar manner, the intersection operation performed at variable nodes simplifies to finding the incoming incoming CTV message of smallest cardinality:
\begin{equation}
\label{intersection_min}
\bigcap\limits_{j \in \mathcal{M}_{\mathtt{c}}} {\mathcal{M}_0^j}  = \mathcal{M}_0^{\mathop {\min}\limits_{j \in \mathcal{M}_\mathtt{c}} j},
\end{equation}
where $\mathcal{M}_{\mathtt{c}}$ is an ordered list containing indices of incoming CTV messages. As a result of \eqref{sumset_union} and \eqref{intersection_min}, the QMBC decoder simplifies to the BEC decoder. That is, a CTV message is a partial erasure if any of the incoming VTC messages is a partial erasure and a VTC message is a partial erasure if the corresponding variable was initially partially erased and all incoming CTV messages are partial erasure. This leads to the BEC density-evolution equation with $\varepsilon =\sum\limits_{j = 1}^s {{\varepsilon _j}}$.
\qed


\section{Proof of Lemma \ref{lem:ML_decoding_all_zeros}}
\label{proof:ML_decoding_all_zeros}

Assume the transmission of a codeword $\boldsymbol{c}$ from a linear code defined by a parity-check matrix $\bf{H}$. Let us denote by ${\boldsymbol{x}}^{(t)}$ ($t = 1,2,...,\prod\limits_{j = 1}^s {\left| {{{\cal{E}}_j}} \right|}$) the GF($q$) words (not necessarily codewords) consistent (see Definition \ref{def:consistent}) with the channel output $\boldsymbol{y}$. That is, any $\boldsymbol{x}^{(t)}$ as input would result in the output $\boldsymbol{y}$, given the partial-erasure index sets $\{{{\mathcal{E}_j}}\}_{j = 1}^s$. An ML decoder fails if and only if there exists $\boldsymbol{x}^{(t)} \ne \boldsymbol{c}$, such that ${\bf{H}} \boldsymbol{x}^{(t)} = \boldsymbol{0}$. Now assume the transmission of the all-zero codeword, and recall that $\mathcal{M}_{c_i}^j = \mathcal{M}_{0}^j + c_i$ (see {Section~\ref{sec:channel_model}). Then each $\boldsymbol{x}^{(t)}$ consistent with the sets $\mathcal{M}_{c_i}^j$ and satisfying ${\bf{H}} \boldsymbol{x}^{(t)} = \boldsymbol{0}$ has a corresponding $\boldsymbol{z}^{(t)}=\boldsymbol{x}^{(t)}-\boldsymbol{c}$ that is consistent with the sets $\mathcal{M}_{0}^j$ and satisfying ${\bf{H}} \boldsymbol{z}^{(t)} = \boldsymbol{0}$. Thus, the probability of decoding failure under ML decoding is independent of the transmitted codeword.

\section{Proof of Lemma \ref{lem:zero_sum}}
\label{proof:zero_sum}

Let us start with the case $m=2$. The elements of the vector $\boldsymbol{a}$ sum to zero if and only if they are the same. Thus, there are $q-1$ vectors with all non-zero elements of length $2$ whose elements sum to zero. As a consequence, there are $(q-1)^2 - (q-1)$ vectors with all non-zero elements whose elements sum to a non-zero field element. Let us move to the $m=3$ case, where we consider a  vector $\tilde{\boldsymbol{a}} = \left(\tilde{a}_1,\tilde{a}_2,\tilde{a}_3\right)$ of $3$ non-zero elements. The equation ${\tilde{a}}_1+ {\tilde{a}}_2 + {\tilde{a}}_3 = 0$ is equivalent to ${\tilde{a}}_1+  {\tilde{a}}_2 =  {\tilde{a}}_3$. As $\tilde{a}_3$ can be any non-zero field element, the number of ways to obtain ${\tilde{a}}_1+ {\tilde{a}}_2 + {\tilde{a}}_3 =0$ is the same as the number of ways to obtain a non-zero sum of ${\tilde{a}}_1+ {\tilde{a}}_2$. According to the previous $m=2$ result, this number is $(q-1)^2 - (q-1)$. Continuing in the same fashion, there are $\sum\limits_{i = 1}^{m - 1} {{{\left( {q - 1} \right)}^i}{{\left( { - 1} \right)}^{m - i - 1}}}$ ways to obtain a zero sum for a random vector of $m$ non-zero elements, $m \ge 2$. Simplifying the sum and normalizing by the number of possible vectors $(q-1)^m$ leads to \eqref{w_1}. The upper bound in \eqref{w_1} is equivalent to ${\left( {1 - q} \right)^{2 - m}} \le 1$, which holds for all $m \ge 2$. This upper bound is sharp, as it is attained with equality for $m=2$.

\end{document}